\documentclass[10pt,twocolumn]{IEEEtran}
\IEEEoverridecommandlockouts

\usepackage{amsthm, amssymb}

\def\bbbb{\mathbb}

\def\ccc{\mathcal}
\def\sinr{\textrm{SINR}}
\def\snr{\textrm{SNR}}

\def\dint{{\rm \ d}}

\newtheoremstyle{slplain}
  {3pt}
  {3pt}
  {\slshape}
  {}
  {\bfseries}
  {.}%
  { }
  {}

\usepackage{color}
\theoremstyle{slplain}

\newtheorem{cor}{Corollary}
\newtheorem{lem}{Lemma}
\newtheorem{pro}{Proposition}

\usepackage{cite}
\usepackage{amsfonts}
\usepackage{multicol,multienum}
\usepackage{subfigure}
\usepackage{multirow}

\ifCLASSINFOpdf
  \usepackage[pdftex]{graphicx}
  \graphicspath{{../pdf/}{../jpeg/}}
  \DeclareGraphicsExtensions{.pdf,.jpeg,.png}
\else
  \usepackage{float}
  \usepackage[dvips]{graphicx}
  \graphicspath{{../eps/}}
  \DeclareGraphicsExtensions{.eps}
\fi

\usepackage[cmex10]{amsmath}
\usepackage{algorithm}
\usepackage{algorithmic}
\usepackage{array}
\usepackage{url}
\usepackage{setspace}
\usepackage{cite}

\begin{document}

\title{Spectrum Sharing for Device-to-Device Communication in Cellular Networks}

\author{
\IEEEauthorblockN{Xingqin Lin, Jeffrey G. Andrews and Amitabha Ghosh}
\thanks{Xingqin Lin and Jeffrey G. Andrews are with Department of Electrical $\&$ Computer Engineering, The University of Texas at Austin, USA. (E-mail: xlin@utexas.edu, jandrews@ece.utexas.edu). Amitabha Ghosh is with Nokia. (E-mail: amitava.ghosh@nsn.com). 

This research was supported by Nokia Siemens Networks and the National Science Foundation grant CIF-1016649. A part of this paper was presented at IEEE Globecom 2013 in Atlanta, GA \cite{lin2013optimal}. 
}
}

\maketitle

\begin{abstract}
This paper addresses two fundamental and interrelated issues in device-to-device (D2D) enhanced cellular networks. The first issue is how D2D users should access spectrum, and we consider two choices: overlay (orthogonal spectrum between D2D and cellular UEs) and underlay (non-orthogonal). The second issue is how D2D users should choose between communicating directly or via the base station, a choice that depends on distance between the potential D2D transmitter and receiver.  We propose a tractable hybrid network model where the positions of mobiles are modeled by random spatial Poisson point process, with which we present a general analytical approach that allows a unified performance evaluation for these questions. Then, we derive analytical rate expressions and apply them to optimize the two D2D spectrum sharing scenarios under a weighted proportional fair utility function. We find that as the proportion of potential D2D mobiles increases, the optimal spectrum partition in the overlay is almost invariant (when D2D mode selection threshold is large) while the optimal spectrum access factor in the underlay decreases. Further, from a coverage perspective, we reveal a tradeoff between the spectrum access factor and the D2D mode selection threshold in the underlay: as more D2D links are allowed (due to a more relaxed mode selection threshold), the network should actually make less spectrum available to them to limit their interference.
\end{abstract}




\section{Introduction}

Device-to-device (D2D) networking allows direct communication between cellular mobiles, thus bypassing the base stations (BS). D2D opens up new opportunities for proximity-based commercial services, particularly social networking applications \cite{Corson2012toward, 3gppD2D}. Other use cases include public safety, local data transfer and data flooding\cite{3gppD2D}. Further, D2D may bring benefits such as increased spectral efficiency, extended cellular coverage, improved energy efficiency and reduced backhaul demand \cite{Fodor2012Design, lin2013overview}. 


\subsection{Related Work and Motivation}

The idea of incorporating D2D communication in cellular networks, or more generally, the concept of hybrid network consisting of both infrastructure-based and ad hoc networks has long been a topic of considerable interest. In earlier studies D2D was mainly proposed for relaying purposes \cite{Ying-Dar2000Multihop, Hongyi2001Integrate}. By allowing radio signals to be relayed by mobiles in cellular networks, it was shown that the coverage and throughput performance can be improved \cite{Ying-Dar2000Multihop, Hongyi2001Integrate}. 
More recently, D2D in cellular networks has been motivated by the trend of proximity-based services and public safety needs \cite{3gppD2D, lin2013overview}. In particular, Qualcomm publicized the necessity of developing a tailored wireless technology to support D2D in cell phones and has also built a D2D demonstration system known as FlashLinQ \cite{wu2010flashlinq, baccelli2012design}. Now D2D is being studied and standardized by the 3rd Generation Partnership Project (3GPP). In particular, potential D2D use cases have been identified in \cite{3gppD2D}; and a new 3GPP study item focusing on the radio access aspect was agreed upon at the December 2012 RAN plenary meeting \cite{3gppD2D2}. Through the most recent 3GPP meetings, initial progress on D2D evaluation methodology and channel models has been made \cite{3gppRAN73}, and broadcast communication is the current focus of D2D study in 3GPP \cite{3gppRAN76}.

In parallel with the standardization effort in industry, basic research is being undertaken to address the many fundamental problems in supporting D2D in cellular networks. One fundamental issue is how to share the spectrum resources  between cellular and D2D communications. Based on the type of spectrum sharing, D2D can be classified into two types: \textit{in-band} and \textit{out-of-band}. In-band refers to D2D using the cellular spectrum, while out-of-band refers to D2D utilizing bands (e.g. 2.4GHz ISM band) other than the cellular band. In-band D2D can be further classified into two categories: \textit{overlay} and \textit{underlay}. Overlay means that cellular and D2D transmitters use orthogonal time/frequency resources, while underlay means that D2D transmitters opportunistically access the time/frequency resources occupied by cellular users.

Existing research relevant to this paper includes spectrum sharing in cognitive radio networks, where secondary cognitive transmitters may access the primary spectrum if the primary transmitters are not active or they do not cause unacceptable interference \cite{akyildiz2008survey}. For example, to protect the primary users, secondary transmissions in \cite{kang2009sensing, song2013optimal}  are regulated by sensing the activities of primary transmissions, while \cite{huang2008opportunistic} imposes stringent secondary access constraints on e.g. collision probability. Multi-antenna techniques are used in \cite{zhang2008exploiting, hamdi2009opportunistic, zhang2011optimal} to minimize secondary interference to primary network. Auction mechanisms are used in \cite{wang2010spectrum, han2011repeated} to control the spectrum access of secondary network.  More recently, the economic aspects of spectrum sharing in cognitive radio networks have gained much interest. For example, \cite{niyato2008competitive} adopts a dynamical game approach to study the spectrum sharing among a primary user and multiple secondary users. Similarly, a three-stage dynamic game is formulated in \cite{duan2011duopoly} to study spectrum leasing and pricing strategies, while  \cite{kalathil2012spectrum} designs incentive schemes for spectrum sharing  with cooperative communication.

Unlike spectrum sharing in cognitive radio networks, D2D spectrum sharing is controlled by the cellular network. How D2D should access the spectrum is a largely open question, though some initial results exist (see e.g. \cite{huang2009spectrum, xu2010effective, yu2011resource, kaufman2013spectrum, ye2013resource}). D2D spectrum sharing is further complicated by D2D \textit{mode selection} which means that a potential D2D pair can switch  between direct and conventional cellular communications \cite{Fodor2012Design, lin2013overview}.   Determining an optimum D2D mode selection threshold -- which we define as the Tx-Rx distance under which D2D communication should occur -- is another objective of this paper.

Note that D2D is different from ad hoc networks whose analysis and design are notoriously difficult (see e.g. \cite{gupta2000capacity, goldsmith2002design, weber2012TC}). A key difference is that D2D networking can be assisted by the cellular network infrastructure which is not available to a typical ad hoc network \cite{lin2013overview, lin2013multicast}. Nevertheless, supporting D2D communication requires a lot of new functionalities \cite{3gppD2D,3gppD2D2, lin2013overview} and significantly complicates the cellular network design. The main objectives of this paper are to provide a tractable baseline model for D2D-enabled cellular networks and to develop a unified analytical framework for the analysis and design of D2D spectrum sharing. This goal is fairly ambitious since D2D spectrum sharing scenarios are quite diverse. The main tool used in this paper is stochastic geometry, particularly the Poisson point processes (PPP). Note that  the PPP model for BS locations has been recently shown to be about as accurate in terms of both signal-to-interference-plus-noise ratio (SINR) distribution and handover rate as the hexagonal grid for a representative urban cellular network \cite{andrews2011tractable, lin2012towards}. 

\subsection{Contributions and Outcomes}

The main contributions and outcomes of this paper are as follows.

\subsubsection{A tractable hybrid network model}
In Section \ref{sec:system}, we introduce a hybrid network model, in which the random and unpredictable spatial positions of mobile users are modeled as a PPP. This model captures many important characteristics of D2D-enabled cellular networks including D2D mode selection, transmit power control and orthogonal scheduling of cellular users within a cell. 

\subsubsection{A unified performance analysis approach}

In Section \ref{sec:unified}, we present a general analytical framework. With this approach, a unified performance analysis is conducted for two D2D spectrum sharing scenarios: overlay and underlay in-band D2D. In particular, we derive analytical rate expressions and apply them to optimize spectrum sharing parameters.

\subsubsection{Design insights}

The following observations are made from the derived analytical and/or numerical results under the model studied in this paper and may be informative for system design. 

\textbf{Overlay vs. underlay.} We evaluate the rate performance in both overlay and underlay scenarios. We observe that D2D mobiles can enjoy much higher data rate than regular cellular mobiles in both scenarios. As for cellular mobiles in the overlay case, their rate performance also improves due to the offloading capability of D2D communication. In contrast, the rate performance of cellular mobiles in the underlay case does not improve or even slightly degrades with D2D communication.\footnote{Note that the underlay study in this paper assumes that D2D randomly accesses the cellular spectrum. With carefully designed dynamic scheduling in the underlay, the rate of cellular mobiles may also increase, and the rate of D2D mobiles may further increase.} This is because cellular mobiles suffer from interference caused by the underlaid D2D transmissions, which offsets the D2D offloading gain. 
From an overall mean-rate (averaged across both cellular and D2D mobiles) perspective, both overlay and underlay provide performance improvement (vs. pure cellular). 

\textbf{D2D mode selection.} We derive the optimal D2D mode selection threshold that minimizes the transmit power. We find that the optimal threshold is inversely proportional to the square root of BS density and monotonically increases with the pathloss exponent. Moreover, it is invariant with the distance distribution of potential D2D pairs. D2D mode selection and spectrum sharing may be jointly optimized from e.g. the rate perspective. From a coverage perspective, we reveal a tradeoff between the D2D spectrum access and mode selection threshold in the underlay case: as more D2D links are allowed (due to a more relaxed mode selection threshold), the network should actually make less spectrum available to them to limit their interference.

The rest of this paper is organized as follows. Section \ref{sec:system} describes the system model and performance metrics. Preliminary analytical results are presented in Section \ref{sec:unified}. Two D2D spectrum sharing scenarios, overlay and underlay, are analyzed and optimized in Sections \ref{sec:overlay} and \ref{sec:underlay}, respectively.  Section \ref{sec:case}  presents a case study of overlay vs. underlay, and is followed by our concluding remarks in Section \ref{sec:concludsions}.

\section{System Model and Performance Metrics}
\label{sec:system}

\subsection{Network Model}
\label{subsec:network}

As shown in Fig. \ref{fig:8}, we consider a hybrid network consisting of both cellular and D2D links and focus on the uplink. The BSs are regularly placed according to a hexagonal grid. Denoting by ${1}/{\lambda_\textrm{b}}$ the area of a hexagonal cell, $\lambda_\textrm{b}$ can be regarded as the average number of BSs per unit area. The transmit user equipments (UE) are randomly distributed and modeled by an independently \textit{marked} PPP 
denoted as
\begin{align}
\tilde{\Phi} = \{ (X_{i},  \delta_{i}, L_{i}, P_{i}  )  \}.  
\end{align}
Here $\{ X_{i} \}$ denote the spatial locations of the UEs. Denote by $\Phi \in \bbbb R^2$ the \textit{unmarked} PPP $\{ X_{i} \}$ with  $\lambda$ being its intensity. $\{ \delta_{i}\}$ denote the types of the UEs and are assumed to be i.i.d. Bernoulli random variables with $\bbbb P(\delta_i=1) = q \in [0,1]$. In particular, UE $i$ is called a \textit{potential} D2D UE\footnote{It is called \textit{potential} D2D UE as a UE with D2D traffic can use either cellular or D2D mode.} if $\delta_i=1$; otherwise, it is called a \textit{cellular} UE. So, $q$ is a simple  indicator of the load of potential D2D traffic.  $\{ L_{i}\}$ denote the lengths of radio links. For notational simplicity, denote by $L_{\textrm{c}}$ (resp. $L_{\textrm{d}}$) the generic random variable for the link length of a typical cellular UE  (resp. potential D2D UE). $\{ P_{i}\}$ denote the transmit powers of UEs. In this paper we use \textit{channel inversion} for power control, i.e., $P_i = L_i^{\alpha}$, where $\alpha > 2$ denotes the pathloss exponent; extension to distance-proportional fractional power control (i.e., $P_i = L_i^{\alpha \epsilon}$ where $\epsilon \in [0,1]$) is straightforward.
Similarly, we use $P_{\textrm{c}}$ and $P_{\textrm{d}}$ to denote the generic random variables for the transmit powers of cellular and potential D2D UEs, respectively.

\begin{figure}
\centering
\includegraphics[width=8cm]{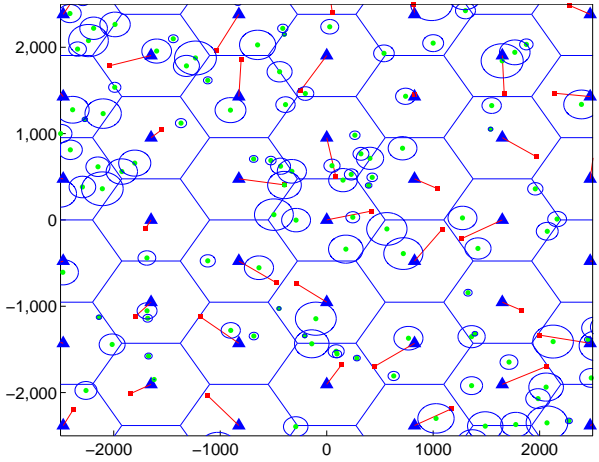}
\caption{A hybrid network consisting of both cellular and D2D links.  Solid triangles, solid squares and dots denote BSs, uplink cellular transmitters and D2D transmitters, respectively. For clarity we omit plotting D2D receivers, each of which is randomly located on the circle centered at the associated D2D transmitter.}
\label{fig:8}
\end{figure}

\textbf{Remark on channel inversion.}
Note that channel inversion in this paper only compensates for the large-scale pathloss. In particular, it does \textit{not} compensate for the small-scale fading. This channel inversion scheme has two advantages: 1) it does not lead to excessively large transmit power when the link is poor (due to the small-scale fading), and 2) the transmitter only needs a long-term statistic (i.e. pathloss) to decide its transmit power, i.e., instantaneous channel state information is not required to be available at the transmitter as small-scale fading is not compensated for. Note that for ease of exposition, we assume in the analysis that the average received power is $1$ due to channel inversion, i.e., $P_i=L_i^{\alpha}$. In other words, $P_i$ should be considered as \textit{virtual} transmit power, and should be scaled appropriately to map to the \textit{actual} transmit power $\tilde{P}_i$, say, $\tilde{P}_i = \rho P_i$, where $\rho$ is the coefficient of proportionality. Normally, $\rho \ll 1$ since the practical transmit power of wireless devices is far less than the pathloss. 

Next, let us introduce the notation $\snr_m$ to denote the average \textit{received} signal power normalized by noise power, i.e.,
\begin{align}
\snr_m = \frac{\tilde{P} L^{-\alpha}    }{\tilde{N}_0 B_w} = \frac{\rho P L^{-\alpha}    }{\tilde{N}_0 B_w}  =  \frac{1}{\rho^{-1} \tilde{N}_0 B_w},
\end{align}
where $\tilde{N}_0$ denotes the one-sided power spectral density of the additive white Gaussian noise, and $B_w$ denotes the channel bandwidth. In the rest of this paper, if the average received power is normalized to $1$, we use $N_0$ to denote the \textit{equivalent} noise power $\rho^{-1} \tilde{N}_0 B_w$. By choosing the operating regime $\snr_m$ (or equivalently, the coefficient $\rho$) appropriately, we can make sure that the UE power constraints are satisfied and thus there is no need to truncate UE transmit power to meet the peak power constraint. We will give more detailed results in Section \ref{subsec:power} to illustrate the above argument.

The potential of D2D will largely depend on the amount of local traffic that may be routed through local direct paths, instead of infrastructure paths. One possible approach to model ``data localization'' would be based on current user traffic statistics. However, it appears very challenging to acquire such traffic data, which is typically owned by operators and contains sensitive and proprietary information. Even if the current traffic data could be obtained from the operators, it might not be too useful, since presumably D2D's availability could change future traffic patterns. For example, once users realize high D2D speeds are possible, more local sharing is likely to occur. So far, no commonly agreed upon D2D distance distribution has appeared in the literature. 
In the absence of such an accepted model, we assume that each potential D2D receiver is randomly and independently placed around its associated potential D2D transmitter with isotropic direction and Rayleigh distributed distance $D$ with
probability density function (PDF) given by
\begin{align}
f_D ( x ) = 2\pi \xi x e^{ -\xi \pi x^2 }, \quad x\geq 0. \label{eq:101}
\end{align}
In other words, the potential D2D receiver is randomly placed around its associated potential D2D transmitter according to a two-dimensional Gaussian distribution, which results in (\ref{eq:101}). A similar Gaussian assumption has also been used in \cite{baccelli2012optimizing}  to analyze the performance of FlashLinQ. The analysis and calculations in this paper can be used to study other D2D distance distributions as well.


In this paper, we consider \textit{distance-based D2D mode selection}: cellular mode is used if $D \geq \mu$; otherwise, D2D mode is selected. If we assume that the received signal power (averaged over fast fading) is only a function of distance and pathloss exponent, distance-based D2D mode selection is equivalent to the average received-signal-power or SNR-based mode selection, to which the results in this paper can be directly applied.

\subsection{Spectrum Sharing}

The spectrum sharing models studied in this paper are 
described as follows.

\subsubsection{Overlay in-band D2D}
The uplink spectrum is divided into two orthogonal portions. A fraction $\eta$ is assigned to D2D communication while the other $1-\eta$ is used for cellular communication. We term $\eta$ \textit{spectrum partition factor} in the overlay.

\subsubsection{Underlay in-band D2D}
We assume that each D2D transmitter uses frequency hopping to randomize its interference to other links. Specifically, we divide the channel into $B$ subchannels. Each D2D transmitter may randomly and independently access $\beta B$ of them, where the factor $\beta \in [0,1]$ measures the aggressiveness of D2D spectrum access. We term $\beta$ \textit{spectrum access factor} in the underlay.


Though we examine D2D spectrum sharing from the frequency domain perspective, it is straightforward to interpret the derived results in this paper from a two-dimensional time-frequency perspective. Take the overlay for example: the equivalent interpretation is that a proportion $\eta$ of the OFDMA (orthogonal frequency-division multiplexing access) resource blocks is assigned to D2D while the remaining resource blocks are used by cellular. 

\subsection{Transmission Scheduling}
\label{subsec:transmission}

Cellular transmitters including cellular UEs and potential D2D UEs in cellular mode  form a PPP $\Phi_{c}$ with intensity $\lambda_{\textrm{c}} = (1-q)\lambda + q \lambda \bbbb P(D \geq \mu) $. We assume an orthogonal multiple access technique and that uplink transmitters are scheduled in a round-robin fashion. It follows that only one uplink transmitter in each macrocell can be active at a time.  Generally speaking,  scheduling cellular transmitters in an  orthogonal manner leads to \textit{dependent} thinning of PPP $\Phi_{\textrm{c}}$.  This makes the analysis intractable and some simplified assumptions are needed (see e.g. \cite{novlan2012analytical}). In this paper, denoting by $\mathcal{A}$ the coverage region of a hexagonal macrocell, we approximate  $\mathcal{A}$ by a disk that has the same area as the hexagonal cell, i.e., $\mathcal{A} = \mathcal{B} (0, R)$ where $\mathcal{B} (x, r)$ denotes the ball centered at $x$ with radius $r$ and $R = \sqrt{\frac{1}{\pi \lambda_{\textrm{b}}}}$. To avoid triviality, we assume $\lambda_{\textrm{c}} \geq \lambda_{\textrm{b}}$, which is reasonable as the uplink transmitter density is usually larger than the BS density.  Further, we assume that the typical active cellular transmitter is uniformly distributed in the coverage region $\mathcal{A}$, and that the locations of cellular interferers form a PPP $\Phi_{\textrm{c,a}}$ with intensity $\lambda_{\textrm{b}}$. For the typical uplink transmission, cellular interferers are located outside the region $\mathcal{A}$. Fig. \ref{fig:5} illustrates the proposed approximate interference analysis for a typical uplink transmission. Due to the use of this approximation, the analytical results about cellular performance derived in this paper are approximations; for notation simplicity, we will present the results as equalities instead of the more cumbersome approximations in the sequel. The approximation will be numerically validated later.

\begin{figure}
\centering
\includegraphics[width=8cm]{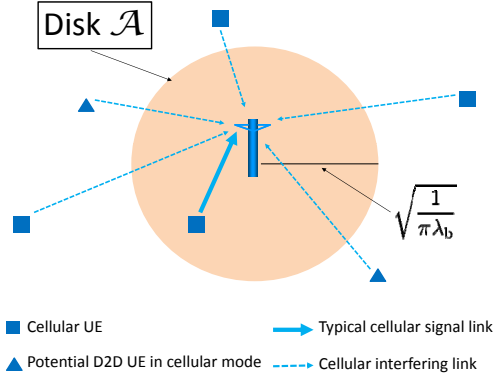}
\caption{An approximate uplink interference analysis. The typical cellular transmitter is uniformly distributed in $\mathcal{A}$, while cellular interferers form a PPP with density $\lambda_{\textrm{b}}$ outside the disk $\mathcal{A}$.}
\label{fig:5}
\end{figure}

As for potential D2D UEs in D2D mode, they form a PPP $\Phi_{\textrm{d}}$ with intensity $\lambda_{\textrm{d}} = q \lambda \bbbb P(D < \mu)$.  For D2D medium access control, we consider a simple spatial Aloha access scheme: in each time slot each potential D2D UE in D2D mode transmits with probability $\kappa$ and is silent with probability $1 - \kappa$; the activity decisions are independently made over different time slots and different transmitters. The study of this simple baseline medium access scheme can serve as a benchmark for future work on more sophisticated D2D scheduling schemes, e.g., carrier sense multiple access (CSMA) or centralized scheduling.

\subsection{Performance Metrics}
\label{subsec:metric}

We will analyze the average rates of cellular and potential D2D UEs, $T_{\textrm{c}}$ and $T_{\textrm{d}}$.  Recall that potential D2D UEs can use either cellular or D2D mode. Denote by $\hat{T}_{\textrm{d}}$ the average rate of potential D2D UEs in D2D mode. Conditioning on using cellular mode, the rate $T_{\textrm{d}}$ of potential D2D UEs equals the  rate $T_{\textrm{c}}$ of cellular UEs; conditioning on using D2D mode, the rate $T_{\textrm{d}}$ of potential D2D UEs equals $\hat{T}_{\textrm{d}}$ by definition. Under the assumed distance-based D2D mode selection, a typical potential D2D UE uses cellular mode with probability $\bbbb P(D \geq \mu )$ and D2D mode with probability $\bbbb P(D < \mu )$. To sum up,  the average rate of potential D2D UEs can be written as
\begin{align}
T_{\textrm{d}}  = \bbbb P(D \geq \mu) \cdot T_{\textrm{c}}  + \bbbb P(D < \mu)  \cdot \hat{T}_{\textrm{d}} .
\end{align}

\section{Preliminary Analysis}
\label{sec:unified}

In this section we present preliminary analytical results, which lay the foundation for the study of overlay and underlay in-band D2D in the next two sections. 

\subsection{A Unified Analytical Approach}

Consider a typical transmitter and receiver pair interfered by  $K$ types of heterogeneous interferers. The set of the $k$-th type of interferers is denoted as $\mathcal{M}_k$. In this paper, we focus on frequency-flat narrowband channels; the results can be readily extended to OFDM-based frequency-selective wideband channels, each of which can conceptually be regarded as a set of parallel frequency-flat narrowband channels.

The baseband received signal $Y_0 [n]$ at the typical receiver located at the origin can be written as
\begin{align}
Y_0 [n] = & \sqrt{P_{0}  L_{0}^{-\alpha} G_0 }  S_{0} [n] \notag \\
&+ \sum_{k =1}^{K} \sum_{i \in \mathcal{M}_{k}} \sqrt{P_{i} \|X_i\|^{-\alpha} G_{i}}    S_{i} [n] + Z [n], 
\label{eq:1}
\end{align}
where $P_{0},  L_{0}, G_0$ and $S_0[n]$ are associated with the typical link and denote the typical link's  transmit power, length, channel fading and unit-variance signal, respectively; 
$P_{i},  \|X_i\|, G_{i}$ and $S_{i} [n]$ are associated with the interfering link from transmitter $i$ to the typical receiver and denote the interfering link's transmit power, length, channel fading and unit-variance signal, respectively; $Z [n]$  is additive white Gaussian noise with constant PSD $\tilde{N}_0$ Watts/Hz. It follows that the received SINR is given by
$
\sinr = \frac{W}{I + \tilde{N}_0  B_w },
$
where signal power $W = P_{0}  L_{0}^{-\alpha}  G_0$, and interference power $I=\sum_{k =1}^{K} \sum_{i \in \mathcal{M}_{k}} P_{i} \|X_i\|^{-\alpha}  G_i$.

In this paper, recall we assume channel inversion, i.e., $P_0  L_0^{-\alpha} = 1$, and thus $W = G_0$. 
For simplicity we consider Rayleigh fading, i.e., $G_0 \sim \textrm{Exp} (1)$, and assume independent fading over space. Then
the following corollary will be particularly useful. 
\begin{cor}
Suppose $\sinr = \frac{W}{I + N_0 }$, where $W \sim \textrm{Exp} (1)$ and $I$ respectively denote the (random) signal and interference powers, and $N_0$ denotes the \textit{equivalent} noise power. If $W$ and $I$ are independent, 
\begin{align}
\bbbb E \left[ \log (1+ \sinr ) \right] 
&= \int_0^\infty   \frac{ e^{ - N_0 x}}{1 + x}  \mathcal{L}_I (x)   \dint  x ,
\label{eq:13}
\end{align}
where $\mathcal{L}_I (s) = \bbbb E[ e^{ -s I } ]$ denotes the Laplace transform of $I$.
\label{thm:2}
\end{cor}
Corollary \ref{thm:2}  follows from a more general result given in \cite{lin2013asilomar} and its proof may also be directly found in \cite{baccelli2009stochastic}. Note that in this paper, interference is not Gaussian but may be considered as \textit{conditionally Gaussian}. Specifically, assuming all the transmitters use Gaussian signaling, the interference is Gaussian conditioned on the fading and node locations in the network. Then treating the interference as noise, we may invoke Shannon's formula to determine the maximum achievable spectral efficiency of a typical link. If the random fading and node locations are furthered averaged out, we arrive at the expression (\ref{eq:13}). Though not optimal in an information-theoretical sense, (\ref{eq:13}) serves as a good performance metric and has been widely adopted in literature \cite{baccelli2009stochastic}.

Next we define the ergodic link spectral efficiency $R$, which combines modulation and coding schemes in the physical layer and multiple access protocols in the medium access control layer, as follows.
\begin{align}
R = \bbbb E \left[ \Delta \cdot \log (1+ \sinr ) \right] ,
\end{align}
where $\Delta$ denotes the time and/or frequency resources accessed by the typical link. For example, in the overlay with 
spectrum partition factor $\eta$,  a typical D2D link with random Aloha access probability $\kappa$ can effectively access $\kappa \eta$ time-frequency resources. We will analyze ergodic link spectral efficiency $R$ in detail in Sections \ref{sec:overlay} and \ref{sec:underlay}.

\subsection{Transmit Power Analysis}
\label{subsec:power}

In this subsection, we analyze the transmit power distributions, particularly $\bbbb E[ P_{\textrm{c}} ]$ and $\bbbb E[ P_{\textrm{d}} ]$, the average transmit powers of cellular and potential D2D UEs. The derived results are not only interesting in its own right but also are extensively used for the analysis of rate performance later. To this end, denote by $\hat{P}_{\textrm{d}}$ the generic random variable for the transmit power of potential D2D UEs in D2D mode.

\begin{lem}
The average transmit powers of a typical cellular UE,  a potential D2D UE  and,  a D2D-mode potential D2D UE are respectively given by
\begin{align}
\bbbb E[ P_{\textrm{c}} ] &= \frac{1}{(1+\frac{\alpha}{2})\pi^{\frac{\alpha}{2}} \lambda_\textrm{b}^{\frac{\alpha}{2}}} \label{eq:12} \\
\bbbb E[ P_{\textrm{d}} ] &= e^{-\xi \pi \mu^2} \bbbb E[ P_{\textrm{c}} ]  +   (\xi \pi)^{-\frac{\alpha}{2}}   \gamma (\frac{\alpha}{2} +1, \xi \pi \mu^2  ) \\
\bbbb E[ \hat{P}_{\textrm{d}} ] &= \frac{(\xi \pi)^{-\frac{\alpha}{2}} }{1 - e^{-\xi \pi \mu^2}} \gamma (\frac{\alpha}{2} +1, \xi \pi \mu^2  ), \label{eq:122}
\end{align}
where $\gamma (s,x) = \int_0^{x  }  z^{s-1} e^{-z} \dint  z$ is the lower incomplete Gamma function.
\label{pro:1}
\end{lem}
\begin{proof}
See Appendix \ref{proof:pro:1}.
\end{proof}

Note that both $\bbbb E[ P_{\textrm{c}} ]$ and $\bbbb E[ P_{\textrm{d}}]$ increase as pathloss exponent $\alpha$ increases and are inversely proportional to the square root of BS density. Next we examine how to choose D2D mode selection threshold $\mu$ to minimize the average transmit power  $\bbbb E[ P_{\textrm{d}}]$  of potential D2D UEs.
\begin{pro}
For any distribution function $f_D(x)$ of the nonnegative random distance $D$, $\bbbb E[ P_{\textrm{d}}]$ is minimized when 
\begin{align}
\mu^{\star} &= (\bbbb E[ P_{\textrm{c}} ])^{\frac{1}{\alpha}} = \left(\frac{1}{1+\frac{\alpha}{2}} \right)^{ \frac{1}{\alpha} }  \sqrt{ \frac{1}{\pi \lambda_\textrm{b}} }. 
\end{align}
\label{pro:2}
\end{pro}
\begin{proof}
See Appendix \ref{proof:pro:2}.
\end{proof}

Prop. \ref{pro:2} shows that $\mu^{\star}$ is only a function of the average transmit power $\bbbb E[ P_{\textrm{c}} ]$ of cellular UEs and is independent of the distribution of $D$; in particular, the Rayleigh assumption made in (\ref{eq:101}) is not needed here. Specifically, $\mu^{\star}$ 
is inversely proportional to the square root of BS density $\lambda_\textrm{b}$, which is intuitive: cellular mode becomes more favorable when more BSs are available. In addition,  $(\frac{1}{1+\frac{\alpha}{2}})^{ \frac{1}{\alpha} }$ monotonically increases as $\alpha$ increases. This implies that $\mu^{\star}$ increases in $\alpha$, agreeing with intuition: local transmission with D2D mode is more favorable  for saving transmit power  when the pathloss exponent increases. 

\begin{table}
\centering
\begin{tabular}{|l||r|} \hline
Density of macrocells $\lambda_{\textrm{b}}$ & $(\pi 500^2)^{-1}$ m$^{-2}$  \\ \hline 
Density of UEs $\lambda$   & $10 \times (\pi 500^2)^{-1}$ m$^{-2}$  \\ \hline 
D2D distance parameter $\xi$   & $10 \times (\pi 500^2)^{-1}$ m$^{-2}$  \\ \hline 
Potential D2D UEs $q$   & $0.2$  \\ \hline 
Pathloss exponent $\alpha$ & $3.5$ \\ \hline
$\snr_m$    & $10$ dB  \\ \hline 
Mode selection threshold $\mu$ & $200$ m \\ \hline 
Aloha access probability $\kappa$ & $1$ \\ \hline 
Spectrum partition factor $\eta$ & $0.2$ \\ \hline  
UE weights $(w_{\textrm{c}},w_{\textrm{d}})$ & $(0.6,0.4)$ \\ \hline
Spectrum access factor $\beta$ & $1$ \\ \hline 
Number of subchannels $B$ & $1$ \\ \hline   
\end{tabular}
\caption{Simulation/Numerical Parameters}
\label{tab:sys:para}
\end{table}

Before ending this section, we give a numerical example in Fig. \ref{fig:00} showing that UE power constraint can be satisfied by choosing the right operating regime $\snr_m$. Throughout this paper, the parameters used in plotting numerical or simulation results are summarized in Table \ref{tab:sys:para} unless otherwise
specified. In Fig. \ref{fig:00}, the cellular peak power $P_{\textrm{c},\max}$ is defined as the minimum transmit power used by a cell-edge cellular transmitter to meet the target $\snr_m$, i.e., $P_{\textrm{c},\max}$ is determined by
$
\snr_m = \frac{P_{\textrm{c},\max} R^{-\alpha}   }{\tilde{N}_0 B_w}.
$
Similar, the D2D peak power $P_{\textrm{d},\max}$ is determined by
$
\snr_m = \frac{P_{\textrm{d},\max} \mu^{-\alpha}   }{\tilde{N}_0 B_w}.
$
The average power of a cellular (resp. D2D) transmitter can be obtained by multiplying (\ref{eq:12}) (resp. (\ref{eq:122})) with the scaling factor $\tilde{N}_0 B_w \cdot \snr_m$. As shown in Fig. \ref{fig:00}, the typical UE power constraint $23$ dBm (i.e. $200$ mW) is well respected even when $\snr_m$=10 dB, a relatively high average received SNR in the uplink. Besides, Fig. \ref{fig:00} also shows that compared to cellular transmitters, D2D transmitters can save about $15$ dB transmit power in achieving the same $\snr_m$ target, demonstrating the energy efficiency of D2D communication. In other words, with the same power budget, the D2D links can enjoy about $15$ dB higher $\snr_m$ than the cellular links.

\begin{figure}
\centering
\includegraphics[width=8cm]{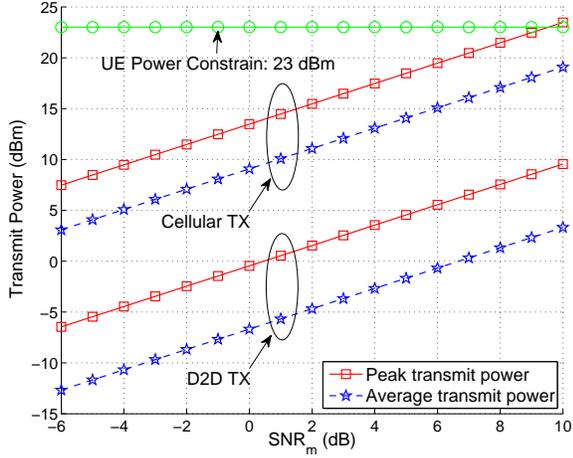}
\caption{UE transmit power versus network operating regime $\snr_m$ with $\tilde{N}_0 = -174$ dBm/Hz and $B_w = 1$ MHz.}
\label{fig:00}
\end{figure}


\section{Analysis of Overlay In-Band D2D}
\label{sec:overlay}

\subsection{Link Spectral Efficiency}

Let us consider a typical D2D link. As the underlying PPP is stationary, without loss of generality we assume that the typical receiver is located at the origin. Note that the analysis carried out for a typical link indicates the spatially averaged performance of the network by Slivnyak's theorem \cite{stoyan1995stochastic}. Henceforth, we focus on characterizing the performance of a typical link, which may be either a D2D or cellular link depending on the context.

With overlay in-band D2D, the interferers are cochannel D2D transmitters. Due to the random Aloha access, the effective interferers  constitute a homogeneous PPP with density $\kappa \lambda_{\textrm{d}}$, a PPP thinned from $\Phi_{\textrm{d}}$ with thinning probability $\kappa$. Denoting this thinned PPP by $\kappa \Phi_{\textrm{d}}$, the interference at the typical D2D receiver is given by
\begin{align}
I_{\textrm{d}} = \sum_{X_i \in \kappa\Phi_{\textrm{d}} \setminus \{o\} } \hat{P}_{\textrm{d},i} G_{i} \|X_{i}\|^{-\alpha} . \notag 
\end{align}
\begin{pro}
With overlay in-band D2D, the complementary cumulative distribution function (CCDF) of the SINR of D2D links is given by
\begin{align}
\bbbb P(\sinr \geq x) =\exp \left(- N_0 x - c  x^{ \frac{2}{\alpha} } \right), \quad x \geq 0, 
\label{eq:ccdf1}
\end{align}
where $c$ is a non-negative constant given by
\begin{align}
c = \frac{ \kappa q(\frac{\lambda}{\xi} - (\frac{\lambda}{\xi}+  \lambda \pi \mu^2) e^{-\xi \pi \mu^2} ) }{\textrm{sinc} (\frac{2}{\alpha}) } . 
\label{eq:cc}
\end{align}
Further, the spectral efficiency $R_{\textrm{d}}$ of D2D links is given by
\begin{align}
R_{\textrm{d}}
&= \kappa \int_0^\infty   \frac{ e^{ - N_0 x}}{1 + x}   e^{ - c  x^{ \frac{2}{\alpha} }  }   \dint x.
\label{eq:100}
\end{align}
\label{pro:3}
\end{pro}
\begin{proof}
This proposition follows by evaluating the Laplace transform of $I_{\textrm{d}}$ and using (\ref{eq:13}). See Appendix \ref{proof:pro:3} for details.
\end{proof}

Note that in Prop. \ref{pro:3}, as $\mu$ increases, $c$ monotonically increases, which in turn results in monotonically decreasing $R_{\textrm{d}}$. This agrees with intuition: the spectral efficiency of D2D links decreases when more potential D2D UEs choose D2D mode (leading to increased interference). 
In particular, 
\begin{align}
R_{\textrm{d},\min} =\lim_{\mu \to \infty} R_{\textrm{d}}
&= \kappa \int_0^\infty   \frac{ e^{ - N_0 x}}{1 + x}  e^{ -\kappa q \frac{\lambda}{\xi} (\textrm{sinc} (\frac{2}{\alpha})  )^{-1}  x^{ \frac{2}{\alpha} }  }   \dint x. \notag 
\end{align}
The typical D2D link experiences the most severe interference in this case and thus has the minimum spectral efficiency. 
On the contrary,
\begin{align}
R_{\textrm{d},\max}
= \lim_{\mu \to 0} R_{\textrm{d}} = \kappa \int_0^\infty   \frac{ 1}{1 + x} e^{ - N_0 x}  \dint x = \kappa e^{N_0} E_1(N_0), \notag 
\end{align}
where  $E_1 (z) = \int_z^{\infty} \frac{e^{-x}}{x} \dint x $ is the exponential integral. The typical D2D link is free of interference in this case and thus has the maximum spectral efficiency.

Now let us consider a typical  uplink. With overlay in-band D2D, the interferers are cellular transmitters in the other cells. The interference at the typical BS is given by
\begin{align}
I_{\textrm{c}} = \sum_{X_i \in \Phi_{\textrm{c,a}} \cap \mathcal{A}^{c} } P_{\textrm{c},i} G_{i} \|X_{i}\|^{-\alpha} .  \notag 
\end{align}

\begin{pro}
With overlay in-band D2D, the CCDF of the SINR of cellular links is given by
\begin{align}
&\bbbb P(\sinr \geq x) = 
\exp \bigg( - N_0 x - 2\pi \lambda_{\textrm{b}} \int_{R}^\infty    \notag \\  
&\quad \bigg(1    - {}_2F_1 (1, \frac{2}{\alpha}; 1+\frac{2}{\alpha}; -  \frac{x}{(R/r)^{\alpha}}   )   \bigg)   r  \dint r  \bigg), \ \ x \geq 0, 
\label{eq:ccdf2}
\end{align}
where ${}_2F_1(a,b;c;x)$ denotes the hypergeometric function, and $R = \sqrt{\frac{1}{\pi \lambda_{\textrm{b}}}}$. Further, the spectral efficiency $R_{\textrm{c}}$ of cellular links is given by
\begin{align}
&R_{\textrm{c}}
 =  \frac{\lambda_{\textrm{b}}}{\lambda_{\textrm{c}} } ( 1 - e^{-\frac{\lambda_{\textrm{c}}}{\lambda_{\textrm{b}}}} ) \int_0^\infty   \frac{ e^{ - N_0 x}}{1 + x} \exp \bigg( - 2\pi \lambda_{\textrm{b}}  \notag \\ 
&\times \int_{R}^\infty \left(1 - {}_2F_1 (1, \frac{2}{\alpha}; 1+\frac{2}{\alpha}; -  \frac{x}{(R/r)^{\alpha}}   )  \right)  r\dint r \bigg)   \dint x , \label{eq:16}
\end{align}
where $\lambda_{\textrm{c}} = (1-q)\lambda + q \lambda e^{-\xi \pi \mu^2}$.
\label{pro:4}
\end{pro}
\begin{proof}
See Appendix \ref{proof:pro:4}.
\end{proof}

Unlike the closed form expression (\ref{eq:ccdf1}) for the CCDF of the SINR of D2D links, the expression (\ref{eq:ccdf2}) for the CCDF of the uplink SINR involves an integration. To have some insights, we consider sparse (i.e. $\lambda_{\textrm{b}} \to 0$) and dense (i.e. $\lambda_{\textrm{b}} \to \infty$) networks in the following corollary. 
\begin{cor}
In a sparse network with $\lambda_{\textrm{b}} \to 0$, the CCDF of the SINR of cellular links is given by 
\begin{align}
\bbbb P(\sinr \geq x) = \exp \left( - \left(N_0 + \frac{4}{\alpha^2 - 4} \right) x    \right) ,  x \geq 0,
\end{align}
In a dense network with $\lambda_{\textrm{b}} \to \infty$, the CCDF of the SINR of cellular links is given by 
\begin{align}
\bbbb P(\sinr \geq x) = \exp \left( - N_0  x  - \frac{ 1 }{2 \textrm{sinc} (\frac{2}{\alpha}  ) }  x^{ \frac{2}{\alpha} }  \right),  x \geq 0.
\end{align}
\label{cor:2}
\end{cor}
\begin{proof}
See Appendix \ref{proof:cro:2}.
\end{proof}

In a sparse network, Corollary \ref{cor:2} implies that interference and noise have the same impact on the SINR coverage performance (in the order sense). From this perspective, we may simply consider interference as an extra source of noise. Thus, the sparse network is noise-limited. In contrast, in a dense network, the impact of interference behaves differently for UEs with different SINR targets. For users with low SINR target (i.e. $x \to 0$), interference has a more pronounced impact on the SINR coverage performance than the noise. The converse is true for users with high SINR target (i.e. $x \to \infty$).

If we denote by $\theta_{\textrm{c}}$ the SINR threshold for successful uplink transmissions and consider the outage probability $\bbbb P(\sinr < \theta_{\textrm{c}})$, Corollary \ref{cor:2} implies that, as $\theta_{\textrm{c}} \to 0$, the outage probability of a sparse network scales as $\Theta ( \theta_{\textrm{c}} )$. In this case, the outage performance is noise-limited. In contrast, the outage probability of a dense network scales as $\Theta ( \theta_{\textrm{c}}^{\frac{2}{\alpha}} )$. In this case, the outage performance is interference-limited.

Before ending this section, we validate the analytical results, particularly the CCDF of the uplink SINR (since we adopt an approximate approach for the analysis on the uplink SINR). As all the major analytical results presented in this paper are functions of SINR, it suffices to  validate the analytical SINR distributions  by simulation rather than repetitively validating each analytical expression which in turn is a function of SINR.  In Fig. \ref{fig:6},  we compare  the analytical uplink SINR CCDF (given in Prop. \ref{pro:4})  to the corresponding empirical distribution obtained from simulation using the hexagonal model. The simulation steps are described as follows.
\begin{enumerate}
\item Place the BSs according to a hexagonal grid in a large area $C$; the area of a hexagon equals $1/\lambda_{\textrm{b}}$.
\item Generate a random Poisson number $N$ with parameter $\lambda_{\textrm{c}}  |C|$.  
\item Generate $N$ points that are uniformly distributed in $C$; these $N$ points represent the cellular transmitters.
\item For each BS, it randomly schedules a cellular transmitter if there is at least one in its coverage region.
\item Determine the transmit power of each scheduled transmitter.
\item Generate independently the fading gains from each scheduled cellular transmitter to each BS.
\item Collect the SINR statistics of the cellular links located in the central hexagonal cells (to avoid boundary effect).
\item Repeat the above steps $10,000$ times. 
\end{enumerate}

Fig. \ref{fig:6} shows that the analytical results match the empirical results fairly well; the small gaps arise as we approximate  the hexagonal model using a PPP model with a guard radius. Recall that in Fig. \ref{fig:6}, $\lambda_{\textrm{b}} = (\pi 500^2)^{-1}$ m$^{-2}$, and thus the network is sparse. In the sparse regime, as established in Corollary 2, $\snr_m$ (or equivalently, $N_0$) has a considerable impact on the uplink SINR CCDF; Fig. \ref{fig:6} confirms the analytical result.

We next compare the  SINR distribution of a typical D2D link (given in Prop. \ref{pro:3})  to the corresponding empirical distribution obtained from Monte Carlo simulation. The results are shown in Fig. \ref{fig:7}, from which we can see that the analytical results closely match the empirical results as in this case no approximation is made in the analysis.

\begin{figure}[ht]
 \centering
 \subfigure[]{
  \includegraphics[width=8cm]{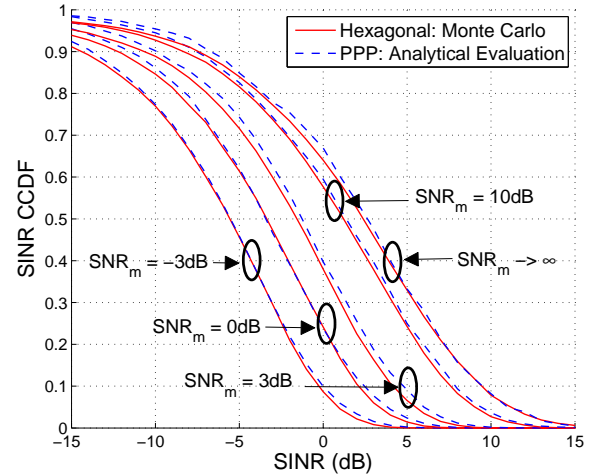}
   \label{fig:6}
   }
 \subfigure[]{
  \includegraphics[width=8.5cm]{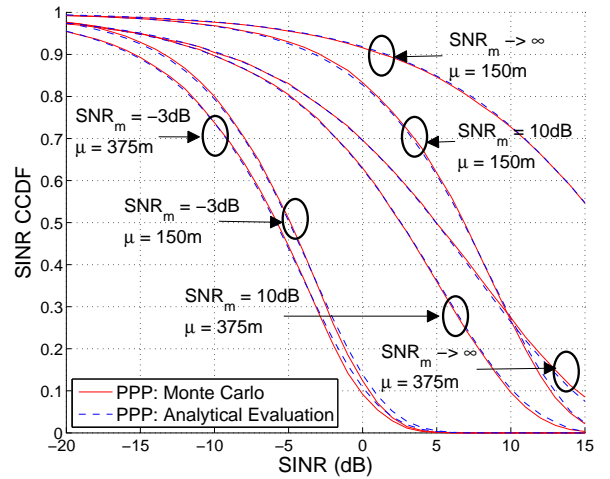}
   \label{fig:7}
   }
 \caption[]{Validation of the analytical SINR CCDF of cellular and D2D links.  Fig. \ref{fig:67}\subref{fig:6} shows the SINR CCDF of cellular links, while Fig. \ref{fig:67}\subref{fig:7} shows the SINR CCDF of D2D links.
  }
 \label{fig:67}
\end{figure}

\subsection{Optimizing Spectrum Partition}
\label{subsec:RopOver}

In this section we study how to choose the optimal spectrum partition factor $\eta^\star$ such that 
\begin{align}
\eta^\star = \arg \max_{\eta \in [0,1]} \quad & u (T_{\textrm{c}}, T_{\textrm{d}} ),
\end{align}
where  $u(T_c,T_d)$ is a utility function that can take different forms under different design metrics. In this paper we use the popular weighted proportional fair function: $w_{\textrm{c}} \log T_{\textrm{c}} + w_{\textrm{d}}\log T_{\textrm{d}}$, where $w_{\textrm{c}}, w_{\textrm{d}} > 0$ are weight factors such that $w_{\textrm{c}} + w_{\textrm{d}} = 1$. 

To optimize the spectrum partition, we first need the rate expressions for $T_{\textrm{c}}$ and $T_{\textrm{d}}$. For a given spectrum partition $\eta$, the (normalized) rate (bit/s/Hz) $T_{\textrm{c}}$ of cellular UEs equals  $R_{\textrm{c}}$ multiplied by the available spectrum resource $1 - \eta$.
In contrast,  the rate $T_{\textrm{d}}$ of potential D2D UEs equals $(1 - \eta) R_{\textrm{c}}$ if cellular mode is used; otherwise, i.e., D2D mode is used, it equals $\eta R_{\textrm{d}}$. To summarize, 
\begin{align}
T_{\textrm{c}} &= (1 - \eta) R_{\textrm{c}}  \notag \\
T_{\textrm{d}} &= (1 - \eta)e^{ -\xi \pi \mu^2 } R_{\textrm{c}} + \eta (1-e^{ -\xi \pi \mu^2 }) R_{\textrm{d}} .
\end{align}

Fig. \ref{fig:9} shows the average rates of cellular and potential D2D UEs as a function of D2D mode selection threshold $\mu$.\footnote{Note that for fair comparison, we normalize the transmit powers of cellular and D2D transmitters by taking into account their transmission bandwidths when plotting numerical results in this paper. For example, in Fig. \ref{fig:9} with $\eta = 0.2$, when $\snr_m = 10$dB, the average received SNR of cellular links and of D2D links are given by $\snr_m + 10\log_{10} \frac{1}{1-\eta} = 10.97$ dB, $\snr_m + 10\log_{10} \frac{1}{\eta} = 16.99$ dB, respectively.} As expected, the average rate of cellular UEs increases as $\mu$ increases. This is because with increasing $\mu$,  less potential D2D UEs choose cellular mode and correspondingly cellular UEs can be scheduled more often. In contrast, the average rate of potential D2D UEs first increases and then decreases as $\mu$ increases. This is because the average rate of potential D2D UEs is co-determined by its cellular-mode rate and D2D-mode rate: cellular-mode rate increases with $\mu$ while D2D-mode rate decreases with $\mu$ (due to the increased intra-tier interference). Fig. \ref{fig:9} also shows that with appropriate choice of $\mu$, potentials D2D UEs can enjoy much higher rate than cellular UEs.  Meanwhile, cellular UEs also benefit from offloading the traffic by D2D communication. 
%


In \cite{lin2013optimal}, we have derived the optimal weighted proportional fair spectrum partition, which reads as follows.
\begin{lem}
The optimal weighted proportional fair spectrum partition $\eta^\star$ is given by
\begin{align}
\eta^\star = 1 - \frac{w_c}{w_c + w_d} \cdot \frac{1}{1 - (e^{\xi \pi \mu^2}-1)^{-1} \frac{R_{\textrm{c}}}{R_{\textrm{d}}} }
\end{align}
if $R_{\textrm{d}} > \frac{w_c+w_d}{w_d} \frac{1}{e^{\xi \pi \mu^2}-1} R_{\textrm{c}} $; otherwise, $\eta^\star=0$. In particular, $\lim_{\mu \to \infty} \eta^\star = w_d$.
\label{pro:7}
\end{lem}

From Lemma \ref{pro:7}, we can see that if $\mu \to \infty$, i.e., potential D2D UEs are restricted to use D2D mode,  the optimal partition $\eta^\star$ converges to $w_{\textrm{d}}$. So the optimal partition $\eta^\star$ simply equals the weight we assign to the potential D2D UEs. 
In Fig. \ref{fig:11}, we plot the utility value vs. $\eta$ under different values of $q$,  the proportion of potential D2D UEs. It can be seen that the optimal $\eta^\star = 0.4 = w_{\textrm{d}}$, which is independent of $q$. This plot validates Lemma \ref{pro:7}.

\begin{figure}[ht]
 \centering
 \subfigure[]{
  \includegraphics[width=8cm]{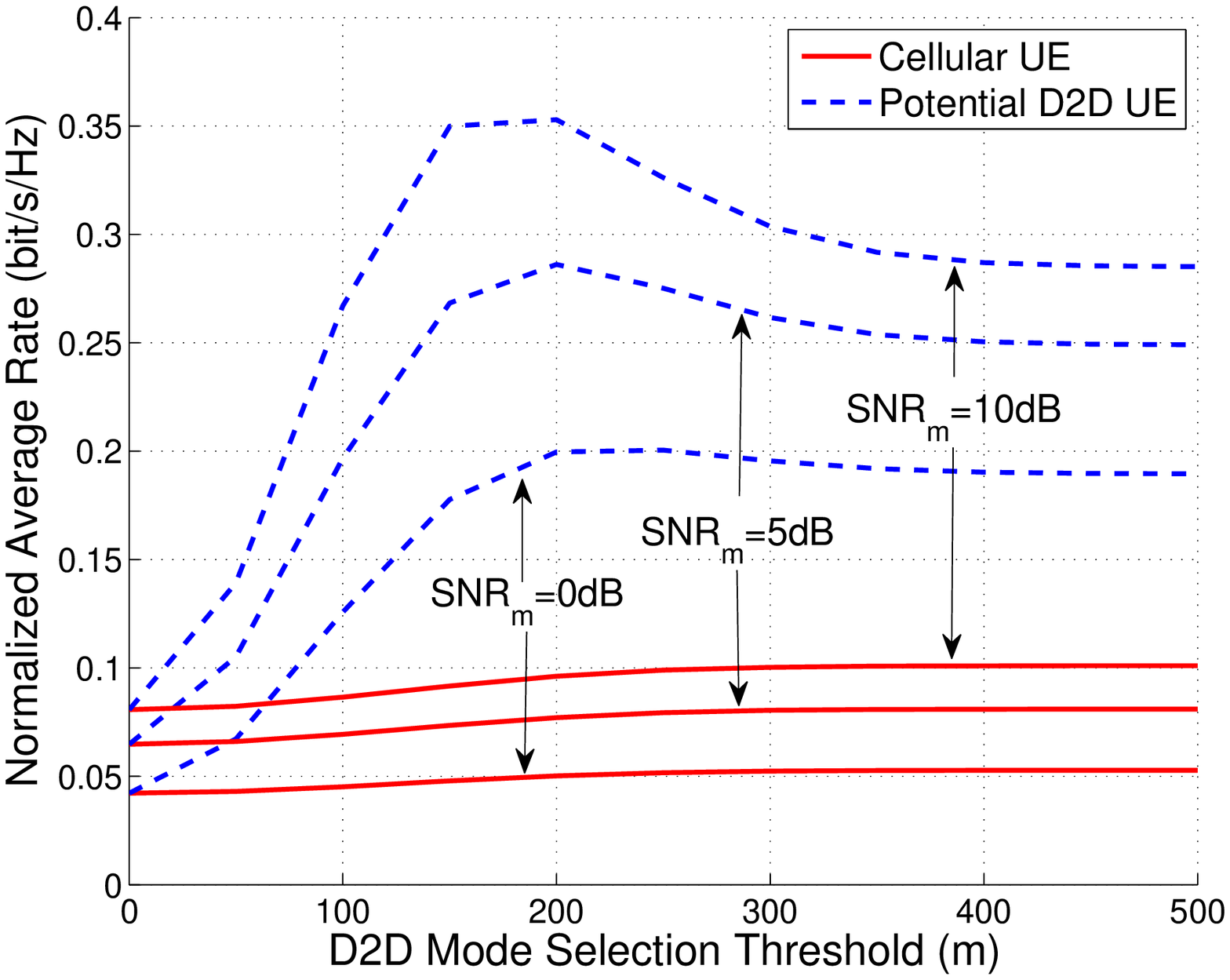}
   \label{fig:9}
   }
 \subfigure[]{
  \includegraphics[width=8cm]{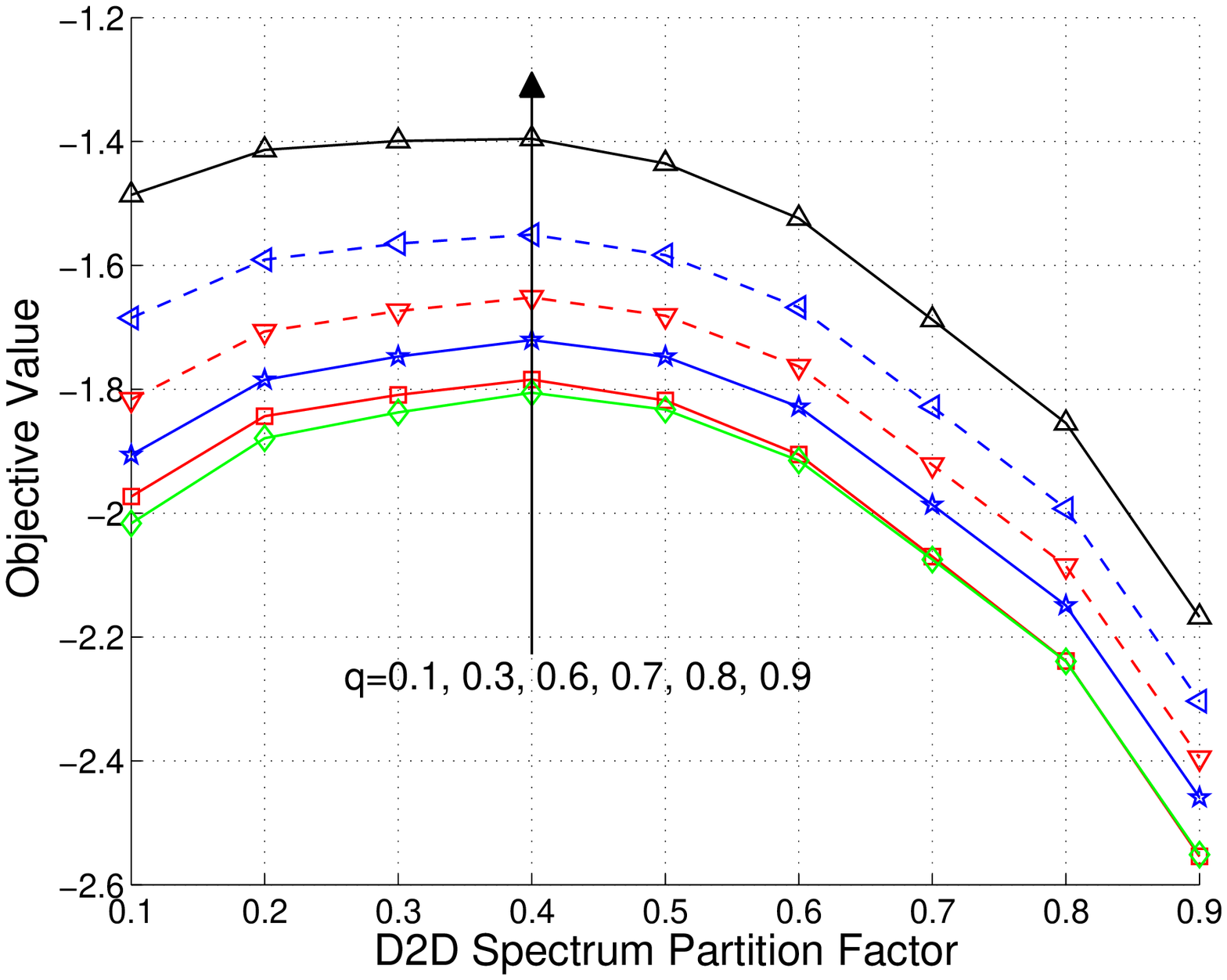}
   \label{fig:11}
   }
 \caption[]{Overlay in-band D2D: Fig. \ref{fig:911}\subref{fig:9} shows the average rates of cellular and potential D2D UEs vs. D2D mode selection threshold. Fig. \ref{fig:911}\subref{fig:11} shows the utility value vs. D2D spectrum partition factor $\eta$ under different values of $q$, the proportion of potential D2D UEs.
}
\label{fig:911}
\end{figure}

Note that the optimal partition $\eta^{\star}$ in Lemma \ref{pro:7} is given for a fixed D2D mode selection threshold $\mu$. With $\eta^{\star} (\mu)$ we can further optimize the objective function $u (T_{\textrm{c}}, T_{\textrm{d}} )$ by choosing the optimal $\mu^{\star}$. With the derived $\eta^{\star} (\mu)$ the objective function $u (T_{\textrm{c}}, T_{\textrm{d}} )$ is only a function of the scalar variable $\mu$. Thus, the optimal $\mu^{\star}$ can be numerically found  efficiently, and the computed $(\mu^{\star}, \eta^{\star} (\mu^{\star}))$  gives the optimal  system design choice. 

\section{Analysis of Underlay In-Band D2D}
\label{sec:underlay}

\subsection{Link Spectral Efficiency}

Considering the random access of D2D in both frequency and time domains, the effective D2D interferers  constitute a homogeneous PPP with density $\kappa \beta \lambda_{\textrm{d}}$, a PPP thinned from $\Phi_{\textrm{d}}$ with thinning probability $\kappa \beta$. Denote this thinned PPP by $\kappa \beta \Phi_{\textrm{d}}$. Considering further the interference from cellular transmitters $\Phi_{\textrm{c,a}}$ with density $\lambda_{\textrm{b}}$, the interference at the typical D2D receiver is given by
\begin{align}
I_{\textrm{d}} = \sum_{X_i \in \kappa \beta \Phi_{\textrm{d}} \setminus \{o\} } \hat{P}_{\textrm{d},i} G_{i} \|X_{i}\|^{-\alpha} +  \sum_{X_i \in \Phi_{\textrm{c,a}} } P_{\textrm{c},i} G_{i} \|X_{i}\|^{-\alpha} . \notag 
\end{align}
\begin{pro}
With underlay in-band D2D, the CCDF of the SINR of D2D links is given by
\begin{align}
&\bbbb P(\sinr \geq x) = \notag \\
&\exp \left( - N_0   x - c \beta x^{ \frac{2}{\alpha} }  -\frac{ 1 }{2 \textrm{sinc} (\frac{2}{\alpha}  ) }  (\beta  x)^{ \frac{2}{\alpha} }  \right), \quad x \geq 0 .
\end{align}
Further, the spectral efficiency $R_{\textrm{d}}$ of D2D links is given by
\begin{align}
R_{\textrm{d}}
&= \kappa \int_0^\infty   \frac{ e^{ - N_0   x}}{1 + x}   \exp \left(- c \beta x^{ \frac{2}{\alpha} }  -\frac{ (\beta  x)^{ \frac{2}{\alpha} } }{2 \textrm{sinc} (\frac{2}{\alpha}  ) }    \right)   \dint x.
\label{eq:161}
\end{align}
\label{pro:8}
\end{pro}
\begin{proof}
See Appendix \ref{proof:pro:8}.
\end{proof}

Now let us consider a typical  uplink. With underlay in-band D2D, the interferers are out-of-cell cellular transmitters $\Phi_{\textrm{c,a}} \cap \mathcal{A}^{c} $ and D2D transmitters in $\kappa \beta \Phi_{\textrm{d}}$. The interference at the typical BS is given by
\begin{align}
I_{\textrm{c}} = \sum_{X_i \in \Phi_{\textrm{c,a}} \cap \mathcal{A}^{c} } P_{\textrm{c},i} G_{i} \|X_{i}\|^{-\alpha} + \sum_{X_i \in \kappa \beta \Phi_{\textrm{d}}  } \hat{P}_{\textrm{d},i} G_{i} \|X_{i}\|^{-\alpha} .  \notag
\end{align}

\begin{pro}
With underlay in-band D2D, the CCDF of the SINR of cellular links is given by
\begin{align}
&\bbbb P(\sinr \geq x) = \exp \bigg( - N_0  x - c \beta^{1-\frac{2}{\alpha}}  x^{ \frac{2}{\alpha} } -  2\pi \lambda_{\textrm{b}}  \notag \\
&\times \int_{R}^\infty \! \! \left(\! 1 \! \!  -  {}_2F_1 (1, \frac{2}{\alpha}; 1+\frac{2}{\alpha}; -  \frac{x}{(R/r)^{\alpha}}   ) \!  \right) \!  r \!\dint r \! \bigg)  ,  x \geq 0 .
\end{align}
Further, the spectral efficiency $R_{\textrm{c}}$ of cellular links is given by
\begin{align}
&R_{\textrm{c}}
= 
\frac{\lambda_{\textrm{b}}}{\lambda_{\textrm{c}} } ( 1 - e^{-\frac{\lambda_{\textrm{c}}}{\lambda_{\textrm{b}}}} )  \int_0^\infty   \frac{ e^{ - N_0  x}}{1 + x}  \exp \bigg(-  c \beta^{1-\frac{2}{\alpha}}   x^{ \frac{2}{\alpha} }  - 2\pi \lambda_{\textrm{b}} \notag \\
&\times \int_{R}^\infty \left(1 - {}_2F_1 (1, \frac{2}{\alpha}; 1+\frac{2}{\alpha}; -  \frac{x}{(R/r)^{\alpha}}   ) \right)  r\dint r \bigg)   \dint x .
\label{eq:1001}
\end{align}
\label{pro:9}
\end{pro}
\begin{proof}
See Appendix \ref{proof:pro:9}.
\end{proof}

Prop. \ref{pro:8} (resp. Prop. \ref{pro:9}) implies that the spectral efficiency $R_{\textrm{d}}$ of D2D links (resp. $R_{\textrm{c}}$ of cellular links) decreases as $\beta$ increases. In other words, with larger $\beta$, the increased D2D interferer density in each subchannel has a more significant impact than the decreased transmit power per subchannel of D2D interferers. To sum up, from the perspective of maximizing the spectral efficiency of either D2D  or cellular links, the design insight here is that underlay in-band D2D should access small bandwidth with high power density rather than spreading the power over large bandwidth. However, small $\beta$ limits the spectrum resource available to the D2D transmissions, which in turn limits the D2D throughput or rate.

\subsection{Optimizing Spectrum Access}

As in the case of overlay, we choose an optimal spectrum access factor $\beta^\star$ in underlay case such that 
\begin{align}
\beta^\star = \arg \max_{\beta \in [0,1]} \quad & w_{\textrm{c}}\log T_{\textrm{c}} + w_{\textrm{d}}\log T_{\textrm{d}} .
\end{align}
To this end, we first need the rate expressions for $T_{\textrm{c}}$ and $T_{\textrm{d}}$. By definition, it is easy to see that
\begin{align}
T_{\textrm{c}} &= R_{\textrm{c}} , \quad
T_{\textrm{d}} &= e^{ -\xi \pi \mu^2 } R_{\textrm{c}} + \beta  (1-e^{ -\xi \pi \mu^2 }) R_{\textrm{d}},
\label{eq:20}
\end{align}
where $R_{\textrm{c}}$ and $R_{\textrm{d}}$ are given in (\ref{eq:1001}) and (\ref{eq:161}), respectively.


Fig. \ref{fig:30} shows the average rates of cellular and potential D2D UEs as a function of $\mu$ in the underlay scenario. 
Recall in the overlay case, the rate of cellular UEs increases with $\mu$ due to D2D offloading gain. In contrast, here the rate of cellular UEs stays almost constant or even slightly decreases with $\mu$. This is because cellular UEs now suffer from the interference caused by the underlaid D2D transmissions; this offsets the offloading gain. Fig. \ref{fig:30} also shows that larger $\beta$ leads to higher rate of potential D2D UEs but lower rate of cellular UEs, which is pretty intuitive.  


Next we optimize the spectrum access. From (\ref{eq:20}), we can see that the spectrum access factor  $\beta$ in the underlay scenario has a much more complicated impact on $T_{\textrm{c}}$ and $T_{\textrm{d}}$ than  $\eta$  does in the overlay scenario. As a result, a closed-form solution for $\beta^\star$ is hard to obtain. Nevertheless,  the optimization problem is of single variable and can be numerically solved. In Fig. \ref{fig:32}, we plot the utility value vs. $\beta$ under different values of $q$,  the proportion of potential D2D UEs. It can be seen that the optimal $\eta^\star$ decreases as $q$ increases.


\begin{figure}[ht]
 \centering
 \subfigure[]{
  \includegraphics[width=8cm]{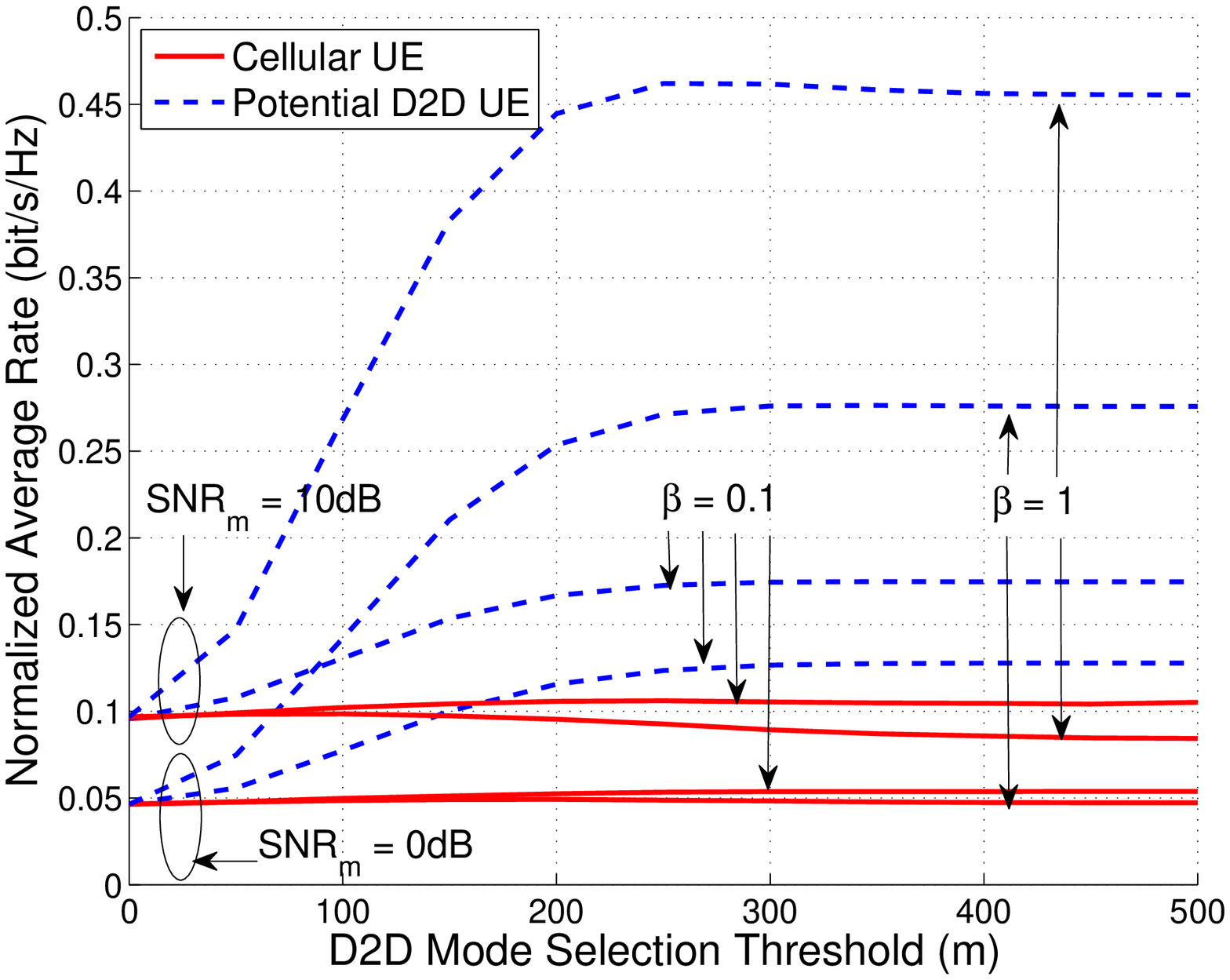}
   \label{fig:30}
   }
 \subfigure[]{
  \includegraphics[width=8cm]{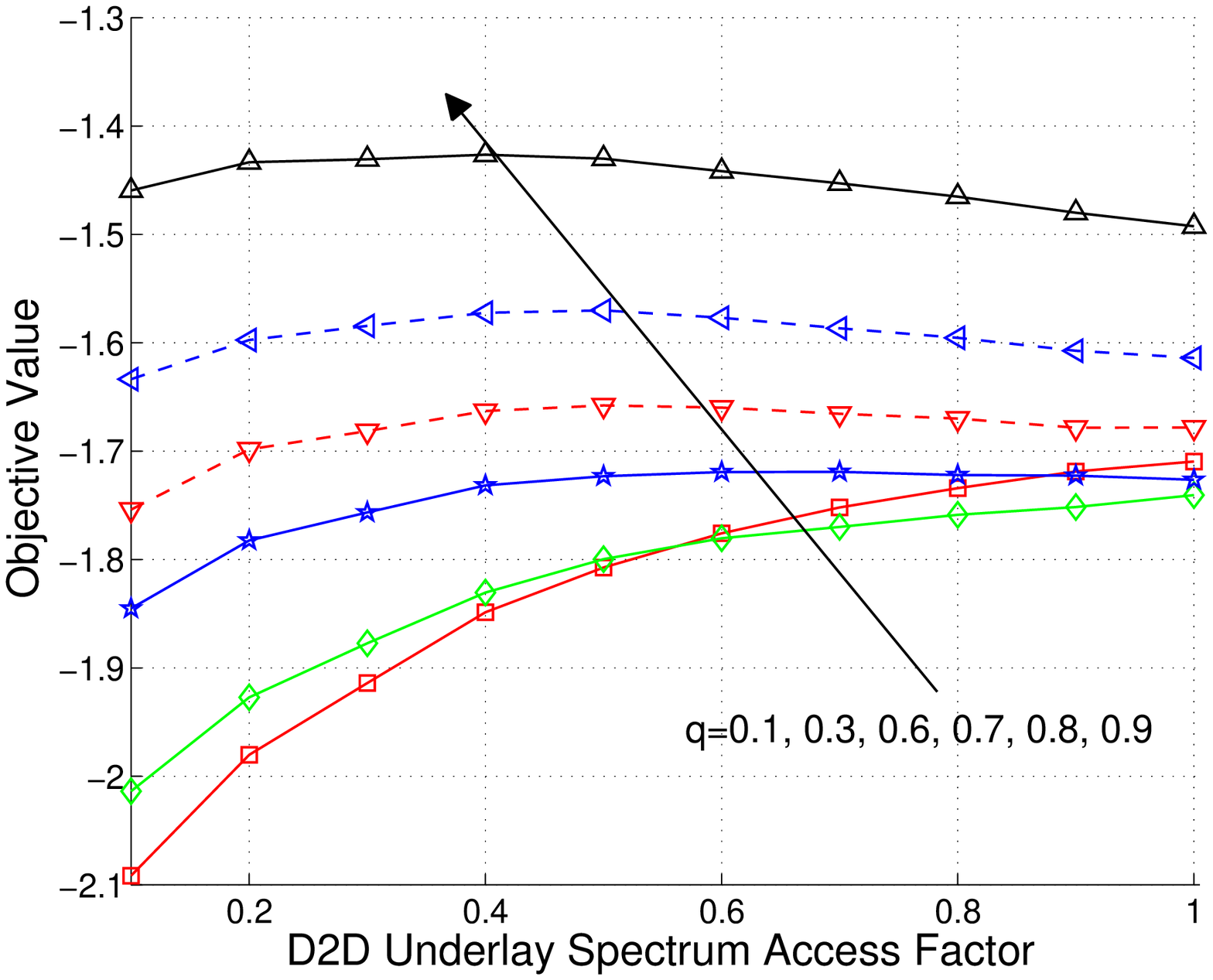}
   \label{fig:32}
   }
 \caption[]{Underlay in-band D2D: Fig. \ref{fig:3032}\subref{fig:30} shows the average rates of cellular and potential D2D UEs vs. D2D mode selection threshold. Fig. \ref{fig:3032}\subref{fig:32} shows the utility value vs. D2D spectrum access factor $\beta$ under different values of $q$, the proportion of potential D2D UEs.
}
\label{fig:3032}
\end{figure}

Thus far, spectrum sharing is optimized from the rate perspective. In practice, D2D spectrum sharing may have to take into account other factors. Take the underlay scenario for example. In order to protect the cellular transmissions, we may have to limit the proportion of the spectrum that can be accessed by D2D. Specifically, assume that D2D transmissions have a target outage probability $\epsilon_{\textrm{d}}$. Then D2D coverage probability must satisfy
\begin{align}
& \bbbb P ( \frac{W}{I_{\textrm{d}}  + N_0 } \geq \theta_{\textrm{d}} ) = \bbbb E[ e^{ - \beta B \theta_{\textrm{d}} (I_{\textrm{d}} + N_0) } ] \geq 1 -  \epsilon_{\textrm{d}}  \quad \Rightarrow    \notag  \\
& \left( N_0 B  \theta_{\textrm{d}}   + \theta_{\textrm{d}}^{ \frac{2}{\alpha} } c (\mu)  \right) \beta +  \frac{ \theta_{\textrm{d}}^{ \frac{2}{\alpha} } }{2 \textrm{sinc} (\frac{2}{\alpha}  ) }  \beta^{\frac{2}{\alpha}}  \leq   \log( \frac{1}{1 - \epsilon_{\textrm{d}}} ) ,
\label{eq:18}
\end{align}
where $\theta_{\textrm{d}}$  is the SINR threshold for successful D2D transmissions, and $c(\mu) = c$ (given in Prop. \ref{pro:3}) monotonically increases as $\mu$ increases. 
Inequality (\ref{eq:18}) reveals the tradeoff between $\beta$ and $\mu$.  In particular,  if each D2D transmission has access to more spectrum, i.e. larger $\beta$, the signal power is spread over wider channel bandwidth and thus the effective SINR gets ``thinner'' in each subchannel. This in turn implies that given link reliability requirement on $\theta_{\textrm{d}}$ and  $\epsilon_{\textrm{d}}$, less cochannel D2D transmissions can be supported, i.e., $\mu$ has to be decreased to make more potential D2D UEs use cellular mode rather than D2D mode. 

Similarly, if cellular transmissions have a target outage probability $\epsilon_{\textrm{c}}$,  the cellular coverage probability must satisfy
\begin{align}
&\bbbb P ( \frac{W}{I_{\textrm{c}} + N_0 } \geq \theta_{\textrm{c}} ) = \bbbb E[ e^{ -\theta_{\textrm{c}} (I_{\textrm{c}} + N_0) } ] \geq 1 -  \epsilon_{\textrm{c}} \quad \Rightarrow   \notag  \\
&  \theta_{\textrm{c}}^{ \frac{2}{\alpha} } c( \mu )  \beta^{1-\frac{2}{\alpha}} \leq    \log( \frac{1}{1 - \epsilon_{\textrm{c}}} ) - N_0 B \theta_{\textrm{c}} - 2\pi \lambda_{\textrm{b}}  \notag \\
& \times \int_{R}^\infty \left(1 - {}_2F_1 (1, \frac{2}{\alpha}; 1+\frac{2}{\alpha}; -  \frac{\theta_{\textrm{c}}}{(R/r)^{\alpha}}   )  \right)  r\dint r .
\label{eq:19}
\end{align}
As in (\ref{eq:18}), a joint constraint on  $\beta$  and $\mu$ is imposed by (\ref{eq:19}); a tradeoff between $\beta$  and $\mu$ exists. Incorporating the constraints (\ref{eq:18}) and (\ref{eq:19}) into the D2D underlay spectrum access optimization problem is an interesting topic for future work.

\section{Overlay vs. Underlay: A Case Study}
\label{sec:case}

In this section we provide a case study to compare the rate performance of overlay with that of underlay. The results are shown in Fig. \ref{fig:4042}, in which the label ``Overall'' indicates the rate performance averaged across both cellular UEs and potential D2D UEs. In Fig. \ref{fig:4042}, the percentage of D2D links equals $q (1-e^{- \xi \pi \mu^2 }) = 0.2   (1 -e^{- \frac{10}{\pi 500^2}  \pi 200^2 } ) = 16\%$. Even with only $16\%$ of the links being D2D links, the overall rate increases remarkably in both overlay and underlay due to the high rates of D2D links. Fig. \ref{fig:4042} further shows that a small $\eta$ (say, $0.3$) in overlay leads to as good rate performance of potential D2D UEs as its counterpart in underlay with a large  $\beta$ (say, $0.9$) because D2D UEs in  overlay are free of cellular interference.

\begin{figure}[ht]
 \centering
 \subfigure[]{
  \includegraphics[width=8cm]{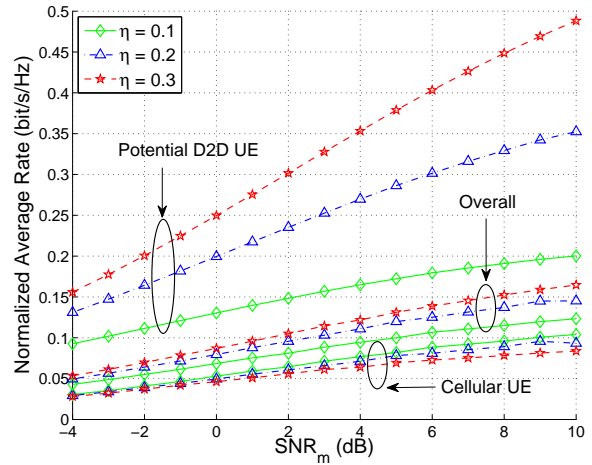}
   \label{fig:40}
   }
 \subfigure[]{
  \includegraphics[width=8cm]{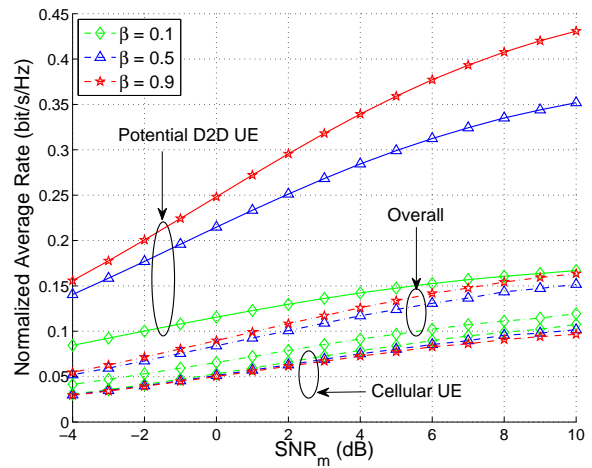}
   \label{fig:42}
   }
 \caption[]{Overlay vs. underlay: Fig. \ref{fig:4042}\subref{fig:40} and Fig. \ref{fig:4042}\subref{fig:42} show the average rate performance of overlay and underlay, respectively. The benchmark curve for \textit{No D2D} is almost indistinguishable from the cellular curve with $\eta = 0.1$ in Fig. \ref{fig:4042}\subref{fig:40} (resp. the cellular curve with $\beta = 0.1$ in Fig. \ref{fig:4042}\subref{fig:42}). For clarity, we do not plot the benchmark curves.
}
\label{fig:4042}
\end{figure}

It can be observed from Fig. \ref{fig:4042} that the rate of potential D2D UEs in overlay increases almost linearly as $\eta$ increases. In contrast,  as $\beta$ increases, the rate of potential D2D UEs in underlay increases in a diminishing way. This is because as $\beta$ increases,  the interference from cellular transmissions and the mutual interference of D2D links increase and thus the received SINR  degrades. Meanwhile, as $\eta$ (resp. $\beta$) increases, the rate of cellular UEs decreases due to the less spectrum resource (resp. more D2D interference). Further, the rate performance of cellular UEs is relatively sensitive to $\eta$ in overlay but is robust to $\beta$ in the underlay.

Recall that the higher the $\snr_m$, the higher the transmit powers. As $\snr_m$ increases from $-4$ dB to $10$ dB,  the rate of cellular UEs in both overlay and underlay and the rate of potential D2D UEs in overlay increase linearly, implying that the transmit powers are not high enough to make the performance interference-limited. In contrast, the rate of potential D2D UEs in underlay increases in a diminishing way, especially when $\snr_m$ exceeds $4$ dB,  implying that the performance is gradually limited by the interference caused by the increased transmit powers.

We summarize the main lessons drawn from the above discussion as follows. In underlay, the rate of potential D2D UEs is more resource-limited; a linear increase in the spectrum resource can generally lead to a linear rate increase. In contrast, the rate of potential D2D UEs is more interference-limited, mainly due to the cochannel cellular interference. As for the rate of cellular UEs, it is  sensitive to the reduction of spectrum in overlay but is more robust to the cochannel D2D interference in underlay.


\section{Conclusions and Future Work}
\label{sec:concludsions}

In this paper we have jointly studied D2D spectrum sharing and mode selection using a hybrid network model and a unified analytical approach. Two scenarios, overlay and underlay, have been investigated. 
Though spectrum sharing has been mainly studied from a rate perspective, we also show in the underlay case how to apply the derived results to study spectrum sharing from a coverage perspective, leading to the discovery of the tradeoff between the underlay D2D spectrum access and mode selection.

The ground cellular network studied in this paper is traditional cellular architecture consisting of only tower-mounted macro BSs. It is interesting to extend this work to (possibly multi-band) heterogeneous networks \cite{andrews2013seven, lin2012modeling} consisting of different types of lower power nodes besides macro BSs. This work can be further extended in a number of ways. At the PHY layer, it would be of interest to study the impact of multiple antenna techniques and/or extend the point-to-point D2D communication to group communication and broadcasting as they are important D2D use cases. At the MAC layer, it would be useful to explore efficient D2D scheduling mechanisms. At the network layer, our work can be used as a stepstone for a wide range of interesting topics including multi-hop and cooperative D2D communications.

\section*{Acknowledgment}

The authors thank Editor Nikos Sagias and the anonymous reviewers for their valuable comments and suggestions.

\appendix

\subsection{Proof of Lemma \ref{pro:1}}
\label{proof:pro:1}

Recall that we adopt an approximate approach on the uplink analysis. In particular, we approximate the coverage region of a hexagonal macrocell as a ball with radius $R = \sqrt{\frac{1}{\pi \lambda_{\textrm{b}}}}$ and assume that the typical active cellular transmitter is uniformly distributed in $\mathcal{A}$. Thus, 
$
\mathbb{P} ( L_\textrm{c} \leq x ) = x^2/R^2, 0 \leq x \leq R.
$
Taking the derivative and using $R = \sqrt{\frac{1}{\pi \lambda_{\textrm{b}}}}$, we obtain that the PDF of the length of  a typical cellular link is
$
f_{L_\textrm{c}} (x) = 2\pi \lambda_{\textrm{b}}  x \cdot \bbbb I_{x  \in [0, \frac{1}{\sqrt{\pi \lambda_{\textrm{b}}}  }  ] }.
$
Then the average transmit power of a cellular transmitter equals
\begin{align}
\bbbb E[ P_{\textrm{c}} ] = \bbbb E[ L^{\alpha}_{\textrm{c}} ] &= \int_{0}^{1/\sqrt{\pi \lambda_\textrm{b}}} 2 \pi \lambda_{\textrm{b}} x^{\alpha+1} \dint x \notag \\
&= \frac{1}{(1+\frac{\alpha}{2})\pi^{\frac{\alpha}{2}} \lambda_\textrm{b}^{\frac{\alpha}{2}}}.  
\label{eq:app:02} 
\end{align}

The PDF of the length of a typical D2D link  is
\[ f_{L_\textrm{d} | D < \mu} (x) = \left\{ \begin{array}{ll}
       \displaystyle \frac{ f_D (x) }{ \bbbb P(D < \mu ) } \!=\! \frac{2 \pi \xi x e^{-\xi \pi x^2}}{ 1 - e^{-\xi \pi \mu^2} } & \mbox{if $x  \in [0, \mu ) $};\\
        0 & \mbox{otherwise}.\end{array} \right. \] 
Correspondingly, its $\alpha$-th moment can be computed as follows:
\begin{align}
\bbbb E[ L^{\alpha}_{\textrm{d}} | D < \mu ] 
=& \frac{1}{  \bbbb P(D < \mu ) } \int_{0}^{\mu}  2 \pi \xi x^{\alpha+1} e^{-\xi \pi x^2}  \dint x \notag \\
=& \frac{(\xi \pi)^{-\frac{\alpha}{2}} }{1 - e^{-\xi \pi \mu^2}} \gamma (\frac{\alpha}{2} +1, \xi \pi \mu^2  ),
\end{align}       
where we have used $\bbbb P(D < \mu ) = 1 - e^{-\xi \pi \mu^2}$ in the last equality. By definition,
\begin{align}
\bbbb E[ P_{\textrm{d}} ] &= \bbbb P( D \geq \mu ) \bbbb E[ L^{\alpha}_{\textrm{c}} ] + (1 - \bbbb P( D \geq \mu )) \bbbb E[ L^{\alpha}_{\textrm{d}} | D < \mu ]  . \notag 
\end{align}
Substituting $\bbbb E[ L^{\alpha}_{\textrm{c}} ]$ and $\bbbb E[ L^{\alpha}_{\textrm{d}} | D < \mu ] $ into the above equation completes the proof.

%

\subsection{Proof of Proposition \ref{pro:2}}
\label{proof:pro:2}

As in the proof of Lemma \ref{pro:1}, the average transmit power of a potential D2D UE equals $\bbbb E[P_{\textrm{c}} ]$ conditioned on cellular mode is used; otherwise, i.e., conditioned on D2D mode is used, it equals $\bbbb E[\hat{P}_{\textrm{d}} ]$. 
Conditioning on D2D mode, the D2D link length $L_{\textrm{d}}$ is distributed as $\frac{1}{\mathbb{P}(D < \mu)} f_D(x), 0\leq L_{\textrm{d}} < \mu$. It follows that $\mathbb{E}[\hat{P}_{\textrm{d}} ] = \mathbb{E}[L^{\alpha}_{\textrm{d}} | D < \mu ] = \int_{0}^{\mu} x^{\alpha} \frac{f_D(x)}{ \mathbb{P}(D < \mu)} \dint x$ and thus
\begin{align}
&\bbbb E[ P_{\textrm{d}} ] =  ( 1 - \int_{0}^{\mu} f_D(x) \dint x )  \cdot  \bbbb E[ P_{\textrm{c}} ] \notag \\
&\quad \quad \quad \quad + \int_{0}^{\mu} f_D(x) \dint x \cdot \int_{0}^{\mu} x^{\alpha} \frac{f_D(x)}{ \mathbb{P}(D < \mu)} \dint x \notag \\
&= ( 1 - \int_{0}^{\mu} f_D(x) \dint x ) \cdot \bbbb E[ P_{\textrm{c}} ] + \int_{0}^{\mu} x^{\alpha} f_D(x) \dint x.  
\end{align}
Taking the derivative of $\bbbb E[ P_{\textrm{d}}]$ with respect to $\mu$,
$
\frac{d }{d \mu }  \bbbb E[ P_{\textrm{d}}] = f_D (\mu) ( \mu^{\alpha}  - \bbbb E[ P_{\textrm{c}}]).  
$
Setting the derivative to zero, we obtain the stationary point $(\bbbb E[ P_{\textrm{c}} ])^{1/\alpha}$. It is easy to see that $\bbbb E[ P_{\textrm{d}}]$ is decreasing when $\mu \in [0,  (\bbbb E[ P_{\textrm{c}} ])^{1/\alpha})$ and is increasing when $\mu \in [(\bbbb E[ P_{\textrm{c}} ])^{1/\alpha}, \infty)$. Hence, $\bbbb E[ P_{\textrm{d}}]$ is minimized at $\mu^{\star} = (\bbbb E[ P_{\textrm{c}} ])^{1/\alpha}$.  The proof is complete by plugging the explicit expression for $\bbbb E[ P_{\textrm{c}} ]$ (given in (\ref{eq:12})).

\subsection{Proof of Proposition \ref{pro:3}}
\label{proof:pro:3}

Consider the conditional Laplace transform
\begin{align}
&\ccc L_{I_{\textrm{d}}} (s) = \bbbb E[ e^{-s { \sum_{X_i \in \Phi_{\textrm{d}} \setminus \{o\} } \hat{P}_{\textrm{d},i} G_{i} \|X_{i}\|^{-\alpha}   } } | o \in \Phi_{\textrm{d}} ] \notag \\
&= \bbbb E^{!o}[ e^{-s \sum_{X_{i} \in \Phi_{\textrm{d}}   } \hat{P}_{\textrm{d},i} G_{i} \|X_{i}\|^{-\alpha}   } ] \notag \\
&=   \bbbb E [ e^{-s \sum_{X_{i} \in \Phi_{\textrm{d}}   } \hat{P}_{\textrm{d},i} G_{i} \|X_{i}\|^{-\alpha}   } ]  \notag \\
&
= \exp \left(  - 2\pi \kappa \lambda_{\textrm{d}} \int_0^{\infty} (1 -  \mathbb{E} [ \exp ( - s \hat{P}_{\textrm{d}} G r^{-\alpha} ) ]  ) r \dint r  \right) \notag \\
& =  \exp \left(  - \frac{ \pi \kappa \lambda_{\textrm{d}} }{\textrm{sinc} (\frac{2}{\alpha}  ) }   \bbbb E [ \hat{P}_{\textrm{d}}^{ \frac{2}{\alpha} }]  s^{ \frac{2}{\alpha} }  \right),
\end{align}
where $\bbbb E^{!o}[\cdot]$ denotes the expectation with respect to the reduced Palm distribution, the third equality follows from Slivnyak's theorem \cite{stoyan1995stochastic}, the fourth equality is due to the probability generating functional of PPP \cite{baccelli2009stochastic}, and we have used $G\sim \textrm{Exp}(1)$ in the last equality. Note that $\lambda_{\textrm{d}} = q\lambda (1 - e^{-\xi \pi \mu^2} ) $, and
from the proof of Lemma \ref{pro:1},
\begin{align}
 \bbbb E [ \hat{P}_{\textrm{d}}^{ \frac{2}{\alpha} }] =  \bbbb E [ L_{\textrm{d}}^{ 2 }] = \frac{1}{\xi \pi} - \frac{ \mu^2e^{-\xi \pi \mu^2}}{1-e^{-\xi \pi \mu^2}}.
\end{align}
Plugging $\lambda_{\textrm{d}}$ and $ \bbbb E [ \hat{P}_{\textrm{d}}^{ \frac{2}{\alpha} }]$ into $\ccc L_{I_{\textrm{d}}} (s)$ yields
$
\ccc L_{I_{\textrm{d}}} (s) = e^{ - c  s^{ \frac{2}{\alpha} }  },
$ where $c$ is given in Prop. \ref{pro:3}. Invoking (\ref{eq:13}) yields the spectral efficiency $R_{\textrm{d}}= \kappa  \bbbb E [ \log (1+\sinr) ] $ of D2D links. 

The CCDF of the SINR of D2D links can be obtained as follows:
\begin{align}
\bbbb P ( \sinr \geq x ) &= \bbbb P \left( G_o \geq x(I_{\textrm{d}} + N_0) \right) \notag \\
&= \bbbb E[ e^{ -x(I_{\textrm{d}} + N_0) } ] = e^{ -x N_0 } \ccc L_{I_{\textrm{d}}} (x) .
\label{eq:comp}
\end{align}
Plugging $\ccc L_{I_{\textrm{d}}} (x)$ into (\ref{eq:comp}) completes the proof.

\subsection{Proof of Proposition \ref{pro:4}}
\label{proof:pro:4}

The spectral efficiency $R_{\textrm{c}}$ of cellular links is given by
\begin{align}
R_{\textrm{c}} &= \bbbb E^o [ \frac{1}{N} \log (1+\sinr) ] \notag \\
&= \bbbb E^o [ \frac{1}{N} ] \cdot \bbbb E [ \log (1+\sinr) ],
\label{eq:app1}
\end{align}
where $N$ is the random number of potential uplink transmitters located in the cell. Due to the PPP assumption, $N$ is a Poisson random variable with parameter $\lambda_{\textrm{c}}/\lambda_{\textrm{b}}$. Denoting by $\tilde{N}$ the number of \textit{other} potential uplink transmitters located in the cell except the one under consideration, i.e., $N = 1+ \tilde{N}$. Thus,
\begin{align}
\bbbb E^o [ \frac{1}{N} ] &= \sum_{n=1}^{\infty} \frac{1}{n} \cdot \bbbb P^o ( N = n ) = \sum_{n=1}^{\infty} \frac{1}{n} \cdot \bbbb P^o ( \tilde{N} = n - 1 )   \notag \\
 &= \sum_{n=1}^{\infty} \frac{1}{n} \cdot  \frac{ (\frac{\lambda_{\textrm{c}}}{\lambda_{\textrm{b}}})^{n-1} }{(n-1)!} e^{-\frac{\lambda_{\textrm{c}}}{\lambda_{\textrm{b}}}} 
= \frac{\lambda_{\textrm{b}}}{\lambda_{\textrm{c}}  } \sum_{n=1}^{\infty} \frac{ (\frac{\lambda_{\textrm{c}}}{\lambda_{\textrm{b}}})^{n} }{n!} e^{-\frac{\lambda_{\textrm{c}}}{\lambda_{\textrm{b}}}} \notag \\ 
&=  \frac{\lambda_{\textrm{b}}}{\lambda_{\textrm{c}} } ( 1 - e^{-\frac{\lambda_{\textrm{c}}}{\lambda_{\textrm{b}}}} ),
\label{eq:app2}
\end{align}
where the third equality is due to Slivnyak's theorem \cite{stoyan1995stochastic}: conditioning on the transmitter under consideration, the other potential uplink transmitters are still PPP distributed and thus $\tilde{N} \sim \textrm{Poisson}( \lambda_{\textrm{c}}/\lambda_{\textrm{b}} )$ under the Palm probability $\mathbb{P}^o$.

Next we calculate $\bbbb E [ \log (1+\sinr) ]$. To this end, we first calculate the Laplace transform
\begin{align}
&\mathcal{L}_{I_{\textrm{c}}} (s) = \bbbb E[ e^{ - s \sum_{X_i \in \Phi_{\textrm{c,a}} \cap \mathcal{A}^{c} } P_{\textrm{c},i} G_{i} \|X_{i}\|^{-\alpha}} ] \notag \\
&= \bbbb E[ e^{ - s \sum_{X_i \in \Phi_{\textrm{c,a}} } P_{\textrm{c},i} G_{i} \|X_{i}\|^{-\alpha} \bbbb I_{ \{ \|X_{i}\| \geq R \} } } ] \notag \\
&= \bbbb E[ \prod_{X_i \in \Phi_{\textrm{c,a}} } e^{ - s  P_{\textrm{c},i} G_{i} \|X_{i}\|^{-\alpha} \bbbb I_{ \{\|X_{i}\| \geq R \} } } ] \notag \\
&= \exp \left( - \int_0^{2\pi} \int_0^\infty 1 - \bbbb E[ e^{ - s  P_{\textrm{c}} G r^{-\alpha} \bbbb I_{ \{r \geq R\} } } ]  \lambda_{\textrm{b}} r\dint r \dint \theta \right)  \notag \\
&= \exp \left( - 2\pi \lambda_{\textrm{b}} \int_{R}^\infty (1 - \bbbb E[ e^{ - s  P_{\textrm{c}} G r^{-\alpha} } ] )  r\dint r \right) \notag \\
&= \exp \left( - 2\pi \lambda_{\textrm{b}} \int_{R}^\infty (1 - \bbbb E[ e^{ - s  L^{\alpha}_{\textrm{c}} G r^{-\alpha} } ] )  r\dint r \right),  
\end{align}
where we have used the probability generating functional of PPP in the fourth equality, and $P_{\textrm{c}} = L^{\alpha}_{\textrm{c}}$ in the last equality. Using that 
$L_{\textrm{c}}$ is distributed as $f_{L_\textrm{c}} (x) = 2\pi \lambda_{\textrm{b}} x \mathbb{I} (x \in [0, 1/\sqrt{\pi \lambda_{\textrm{b}}}] )$, and $G \sim \textrm{Exp} (1)$, 
\begin{align}
\bbbb E_{G,L_\textrm{c}}[ e^{ - x  L^{\alpha}_{\textrm{c}} G r^{-\alpha} } ] &= \bbbb E_{L_\textrm{c}} \left[ \frac{1}{ 1 +  x   L^{\alpha}_{\textrm{c}} r^{-\alpha} }  \right] \notag \\
&=  2\pi \lambda_{\textrm{b}}\int_0^{\sqrt{\frac{1}{\pi \lambda_{\textrm{b}}}}} \frac{t}{ 1 + x r^{-\alpha} t^{\alpha} }  \dint t 
\notag \\
&= {}_2F_1 (1, \frac{2}{\alpha}; 1+\frac{2}{\alpha}; -  \frac{x}{(r\sqrt{\pi \lambda_{\textrm{b}}})^{\alpha}}   ) .\notag 
\end{align}
Plugging $\mathcal{L}_{I_{\textrm{c}}} (s)$ into (\ref{eq:13}) yields 
\begin{align}
&\bbbb E [  \log (1+\sinr) ] = \int_0^\infty   \frac{ e^{ - N_0 x}}{1 + x}  \exp \bigg( - 2\pi \lambda_{\textrm{b}} \notag \\
&\times \int_{R}^\infty \!\! (1 - {}_2F_1 (1, \frac{2}{\alpha}; 1+\frac{2}{\alpha}; -  \frac{x}{(r\sqrt{\pi \lambda_{\textrm{b}}})^{\alpha}}   ) )  r\dint r \bigg)   \dint x .
\label{eq:app3}
\end{align}
Plugging (\ref{eq:app2}) and (\ref{eq:app3}) into (\ref{eq:app1}) gives the spectral efficiency $R_{\textrm{c}}$ of cellular links. The CCDF of the SINR of cellular links can be similarly obtained as in (\ref{eq:comp}).

\subsection{Proof of Corollary \ref{cor:2}}
\label{proof:cro:2}

For a sparse cellular network with small $\lambda_{\textrm{b}}$, $R$ is large and thus for $r\in[R,\infty )$,
\begin{align}
1 - \bbbb E[ e^{ - x   G L^{\alpha}_{\textrm{c}} r^{-\alpha} } ] &\approx x r^{-\alpha}   \bbbb E[ G L^{\alpha}_{\textrm{c}}  ] 
=  x r^{-\alpha}    \bbbb E[ G ] \bbbb E [ P_{\textrm{c}} ] \notag \\
&=   \frac{x r^{-\alpha}   }{ (\frac{\alpha}{2}+1) (\pi \lambda_{\textrm{b}} )^{\frac{\alpha}{2}} },
\end{align}
where we have used the approximation $1 - e^{-y} \approx y$ for small value $y$, the independence of fading $G$ and link length $L_{\textrm{c}}$ and $\bbbb E [ P_{\textrm{c}} ] = \bbbb E[ L^{\alpha}_{\textrm{c}}  ]$ in the first equality, and have plugged in $\bbbb E[ G ] = 1$ and $\bbbb E[ L^{\alpha}_{\textrm{c}}  ] = { 1 }/{ (\frac{\alpha}{2}+1) (\pi \lambda_{\textrm{b}} )^{\frac{\alpha}{2}} }$ (c.f. Lemma \ref{pro:1}) in the last equality.
Accordingly, the Laplace transform of the uplink interference is given by
\begin{align}
\mathcal{L}_{I_{\textrm{c}}} (s)   &= \exp \left( - 2\pi \lambda_{\textrm{b}} \int_{R}^\infty (1 - \bbbb E[ e^{ - s   G L^{\alpha}_{\textrm{c}} r^{-\alpha} } ] )  r\dint r \right) \notag \\
&\approx \exp \left( - 2\pi \lambda_{\textrm{b}} s \int_{R}^\infty \frac{r^{-\alpha+1}  }{ (\frac{\alpha}{2}+1) (\pi \lambda_{\textrm{b}} )^{\frac{\alpha}{2}} }  \dint r \right) \notag \\
&= \exp \left( - \frac{4}{\alpha^2 - 4} s \right). 
\end{align}

For a dense network with large $\lambda_{\textrm{b}}$, $R$ is small and thus
\begin{align}
\mathcal{L}_{I_{\textrm{c}}} (s)  &\approx \exp \left( - 2\pi \lambda_{\textrm{b}} \int_{0}^\infty (1 - \bbbb E[ e^{ - s   G L^{\alpha}_{\textrm{c}} r^{-\alpha} } ] )  r\dint r \right) \notag \\
&= \exp \left( - \frac{ \pi \lambda_{\textrm{b}} }{\textrm{sinc} (\frac{2}{\alpha}  ) }  \bbbb E [ {P}_{\textrm{c}}^{ \frac{2}{\alpha} }]  s^{ \frac{2}{\alpha} }  \right) \notag \\
&= \exp \left( - \frac{ 1 }{2 \textrm{sinc} (\frac{2}{\alpha}  ) }   s^{ \frac{2}{\alpha} }  \right) , \label{eq:17}
\end{align}
where we have plugged in $\bbbb E [ {P}_{\textrm{c}}^{ \frac{2}{\alpha} }] = \bbbb E [ L_{\textrm{c}}^{ 2 }] = \frac{1}{2\pi \lambda_{\textrm{b}}}$ (c.f. (\ref{eq:app:02})) in the last equality. Combining the above asymptotic results with Prop. \ref{pro:4} completes the proof.

\subsection{Proof of Proposition \ref{pro:8}}
\label{proof:pro:8}
Note that the average D2D transmit power $\bbbb E[\hat{P}_{\textrm{d}}]$ here is only $1/\beta B$ of the one given in Lemma \ref{pro:1} as each D2D transmitter accesses $\beta B$ subchannels and needs to split its power accordingly, i.e.,
$
\hat{P}_{\textrm{d}} = \frac{1}{\beta B} L_{\textrm{d}}^{\alpha}.
$
Hence, 
$\bbbb E[\hat{P}^{\frac{2}{\alpha}}_{\textrm{d}}] = (\frac{1}{\beta B})^{\frac{2}{\alpha}} \bbbb E[ L_{\textrm{d}}^2 ]$. The corresponding spectral efficiency (normalized by bandwidth $B$) is given by
\begin{align}
R_{\textrm{d}}  = & \kappa  \bbbb E [ \log (1+\sinr) ] 
=   \kappa \int_0^\infty   \frac{ e^{ - N_0  x}}{1 + x}  \mathcal{L}_{I_{\textrm{d}}} (\beta B x)   \dint x  \notag \\
= &  \kappa \int_0^\infty   \frac{ e^{ - N_0  x}}{1 + x}    \exp \bigg( - \frac{ \pi \kappa \lambda_{\textrm{d}} }{\textrm{sinc} (\frac{2}{\alpha}  ) }  \bbbb E [ \hat{P}_{\textrm{d}}^{ \frac{2}{\alpha} }]  (\beta B x)^{ \frac{2}{\alpha} } \bigg)   \notag \\   
&\quad \times \exp \bigg(  - \frac{ \pi \lambda_{\textrm{b}} }{\textrm{sinc} (\frac{2}{\alpha}  ) }  \bbbb E [ P_{\textrm{c}}^{ \frac{2}{\alpha} }]  (\beta B x)^{ \frac{2}{\alpha} } \bigg)   \dint x .
\label{eq:app:03}
\end{align}
Here $\lambda_{\textrm{d}} = \beta \cdot q\lambda (1 - e^{-\xi \pi \mu^2} ) $, the density of D2D transmitters  that is ``seen'' from each subchannel, and
from the proof in Appendix \ref{proof:pro:1},
\begin{align}
 \bbbb E [ \hat{P}_{\textrm{d}}^{ \frac{2}{\alpha} }] &= (\frac{1}{\beta B})^{\frac{2}{\alpha}} \bbbb E [ L_{\textrm{d}}^{ 2 }] = (\frac{1}{\beta B})^{\frac{2}{\alpha}} \left( \frac{1}{\xi \pi} - \frac{ \mu^2e^{-\xi \pi \mu^2}}{1-e^{-\xi \pi \mu^2}} \right)  \notag \\
\bbbb E [ P_{\textrm{c}}^{ \frac{2}{\alpha} }] &= (\frac{1}{B})^{\frac{2}{\alpha}} \bbbb E [ L_{\textrm{c}}^{ 2 }] = (\frac{1}{B})^{\frac{2}{\alpha}} \frac{1}{2 \pi \lambda_{\textrm{b}}}. \notag
\end{align}
Plugging $\lambda_{\textrm{d}}$, $ \bbbb E [ \hat{P}_{\textrm{d}}^{ \frac{2}{\alpha} }]$ and $\bbbb E [ P_{\textrm{c}}^{ \frac{2}{\alpha} }]$ into (\ref{eq:app:03}) establishes the expression for the D2D link spectral efficiency. The CCDF of the SINR of D2D links can be similarly obtained as in (\ref{eq:comp}).

\subsection{Proof of Proposition \ref{pro:9}}
\label{proof:pro:9}

As in the proof of Prop. \ref{pro:8}, the Laplace transform of the interference can be calculated as follows.
\begin{align}
&\mathcal{L}_{I_{\textrm{c}}}  ( s ) =  \exp \bigg( - \frac{ \pi \kappa \lambda_{\textrm{d}} }{\textrm{sinc} (\frac{2}{\alpha}  ) }  \bbbb E [ \hat{P}_{\textrm{d}}^{ \frac{2}{\alpha} }]  s^{ \frac{2}{\alpha} } - 2\pi \lambda_{\textrm{b}}   \notag \\
& \times  \int_{R}^\infty \left(1 - {}_2F_1 (1, \frac{2}{\alpha}; 1+\frac{2}{\alpha}; -  \frac{s/B}{(r\sqrt{\pi \lambda_{\textrm{b}}})^{\alpha}}   ) \right)  r\dint r \bigg) .
\label{eq:app:04}
\end{align}
Then as in the proof of Prop. \ref{pro:4}, the spectral efficiency $R_{\textrm{c}}$ of cellular links is given by
\begin{align}
R_{\textrm{c}} &= \bbbb E^o [ \frac{1}{N} \log (1+\sinr) ] \notag \\
&
= \frac{\lambda_{\textrm{b}}}{\lambda_{\textrm{c}}} ( 1 - e^{-\frac{\lambda_{\textrm{c}}}{\lambda_{\textrm{b}}}} ) \int_0^{\infty} \frac{e^{-N_0 x}}{1+x}  \mathcal{L}_{I_{\textrm{c}}} ( B x ) \dint x .
\end{align}
Plugging $\mathcal{L}_{I_{\textrm{c}}} ( s )$ (\ref{eq:app:04}) into the above equation establishes the expression for the cellular link spectral efficiency. The CCDF of the SINR of D2D links can be similarly  obtained  as in (\ref{eq:comp}).

\bibliographystyle{IEEEtran}
\bibliography{IEEEabrv,Reference}


\end{document}